\documentclass[runningheads]{llncs}
\newif\ifanon\anonfalse
\newif\iffinal\finaltrue
\newif\ifcameraready\camerareadyfalse

\usepackage[T1]{fontenc}

\usepackage{graphicx}
\usepackage{hyperref}
\usepackage{color}
\usepackage[table]{xcolor}
\iffinal
  \usepackage[disable]{todo}
\else
\usepackage{todo}
\usepackage{showlabels}
\fi

\usepackage{amsmath}
\usepackage{amsfonts}
\usepackage{wasysym}
\usepackage{mathtools}
\usepackage{amssymb}
\usepackage[inline]{enumitem}
\usepackage{algorithm}
\usepackage[noend]{algpseudocode}
\usepackage{multirow}
\usepackage{nicematrix}
\usepackage{subcaption}
\usepackage{cleveref}

\usepackage{tikz}
\usetikzlibrary{automata, positioning, chains, fit, shapes, patterns, matrix}

\usepackage{pgfplots}
\usepackage{pgfplotstable}
\usepackage{csvsimple}

\newcommand{\calO}{\mathcal{O}}
\newcommand{\notleftrightarrow}
  {\mathrel{\ooalign{$\leftrightarrow$\cr\hidewidth$/$\hidewidth}}}
\newcommand{\integers}{\mathbb{N}}
\newcommand{\bigo}{\mathcal{O}}
\newcommand*\interv[2]{\ensuremath{[{#1}\mathrel{{.}{.}}\nobreak{#2}]}}
\newcommand*\suffx[2]{\ensuremath{\prescript{#2}{}{#1}}}

\newcommand{\lang}{\mathcal{L}}
\newcommand{\aut}{\mathcal{A}}
\newcommand{\Rat}{\mathrm{Reg}}
\DeclareMathOperator{\lequiv}{\equiv_L}
\DeclareMathOperator{\notlequiv}{{\not\equiv}_L}
\newcommand{\lqspace}{\ensuremath{\Sigma^* / \lequiv}}

\DeclareMathOperator{\tequiv}{\equiv_\mathfrak{T}}
\DeclareMathOperator{\nottequiv}{{\not\equiv}_\mathfrak{T}}
\newcommand{\vars}{\mathcal{V}}
\newcommand{\propfor}{\mathcal{F}_0}

\newcommand{\rts}{\mathcal{R}}
\newcommand{\trans}{\ensuremath{\hookrightarrow}}
\newcommand{\post}{\mathrm{Post}}
\newcommand{\pre}{\mathrm{Pre}}

\newcommand{\tgt}{\ensuremath{\mathcal{T}}}
\newcommand{\tgtsep}{\mathcal{T}_{\text{sep}}}
\newcommand{\tgtind}{\ensuremath{\mathcal{T}_{\text{ind}}}}
\newcommand{\indpairs}{\ensuremath{\mathrm{IP}}}

\newcommand{\MEM}{\ensuremath{\mathtt{MEM}}}
\newcommand{\EQ}{\ensuremath{\mathtt{EQ}}}
\newcommand{\VAL}{\ensuremath{\mathtt{VAL}}}

\newcommand{\toff}{\Circle}
\newcommand{\ton}{\CIRCLE}

\newcommand{\lstar}{\ensuremath{L^{*}}}

\newcommand{\linda}{\ensuremath{L^{\text{IndA}}}}
\newcommand{\obstab}{\mathfrak{T}}
\newcommand{\cells}{\mathfrak{C}}
\def\rouge[#1]{\textcolor{purple}{#1}}
\def\bleu[#1]{\textcolor{blue}{#1}}
\def\gris[#1]{\textcolor{gray}{#1}}
\def\vert[#1]{\textcolor{teal}{#1}}
\newcommand{\val}{\nu}
\newcommand{\clauses}{\mathcal{C}}

\newcommand{\core}{\mathcal{K}}
\newcommand{\hyp}{\ensuremath{\mathcal{H}}}
\newcommand{\hypm}{\ensuremath{\mathcal{H}_{m}}}
\makeatletter
\@ifclassloaded{beamer}{}{}
\makeatother
\newcommand*{\cbasis}[1]{\ensuremath{(\mathtt{basis}_{#1})}}
\newcommand*{\creach}[1]{\ensuremath{(\mathtt{reach}_{#1})}}
\newcommand*{\ccong}[1]{\ensuremath{(\mathtt{congr}_{#1})}}
\newcommand*{\cdet}[1]{\ensuremath{(\mathtt{det}_{#1})}}
\newcommand*{\cclos}[1]{\ensuremath{(\mathtt{clos}_{#1})}}

\newcommand*{\csucc}[1]{\ensuremath{(\mathtt{succ}_{#1})}}
\newcommand*{\cind}[1]{\ensuremath{(\mathtt{ind}_{#1})}}
\newcommand*{\csharp}[1]{\ensuremath{(\mathtt{sharp}_{#1})}}
\newcommand*\repr[2]{\ensuremath{\sigma_{#2}({#1})}}
\newcommand*\pcomp[3][]{c_{#2,#3}^{#1}}
\newcommand*\agree[2]{\alpha_{#1}({#2})}
\newcommand*\agreetd[3]{\alpha_{#1}({#2,#3})}
\newcommand{\iter}[2]{{#1}^{(#2)}}
\newcommand{\Rq}[1]{\textit{(Rq#1)}}

\newcommand{\defn}[1]{\emph{#1}}

\newcommand{\nospace}[1]{\makebox[0pt][l]{$#1$}}
\newcommand{\idmat}{IdMAT}

\newcommand{\appref}[2]{\ifcameraready\cite[Apppendix #1]{IdMATTR}\else Appendix~\ref{#2}\fi}

\begin{document}
\title{Automata Learning with an Incomplete but Inductive Teacher}
%
%
\ifanon
  \author{Anonymous Authors}
  \authorrunning{}
\else
\author{Daniel Stan\inst{1,2}\orcidID{0000-0002-4723-5742} \and
Adrien Pommellet\inst{1}\orcidID{0000-0001-5530-152X} \and
Juliette Jacquot\inst{1}}
\authorrunning{D.~Stan, A.~Pommellet, J.~Jacquot}
\fi
%
\ifanon
\else
\institute{Laboratoire de Recherche de l’EPITA, 14-16 Rue Voltaire, Le Kremlin-Bicêtre,
France\\
\email{\{daniel.stan,adrien.pommellet,juliette.jacquot\}@epita.fr}\\
\url{www.lre.epita.fr} \and
ICube, UMR 7357, Université de Strasbourg, Strasbourg, France
}
\fi
\maketitle              
\begin{abstract}
  \defn{Active automata learning} (AAL) under a \defn{Minimally Adequate 
  Teacher} (MAT) has been successfully used to infer a regular language
  through membership and equivalence queries.
  This language might not be fully characterized:
  we are thus interested in finding \emph{any} solution in a
  \defn{target class} of possibly many regular languages.
  Some problems such as \defn{regular language separation} or \defn{inductive 	
  invariant synthesis} in the context of \defn{regular model checking} (RMC) 
  may indeed admit more than one answer.
  We therefore introduce IdMAT:
  a new teacher formalism answering queries with 
  respect to any language in the target class, all at once.
  Such a teacher—tailored towards invariant synthesis—might provide 
  \defn{incomplete} ``don't know'' answers, but also
  \defn{inductive} facts of the form ``if $w_1$ is accepted, so is $w_2$''.
  We pair IdMAT with a novel AAL algorithm $\linda$ that
  \begin{enumerate*}
    \item encodes all uncertainties as a unique SAT instance and does not fork,
    \item leverages incremental SAT solving and UNSAT core analysis, and
    \item handles counterexamples—of the simple or inductive 
    type—in a frugal manner inspired by the Rivest-Schapire 
    refinement technique.
  \end{enumerate*}
  We finally evaluate a prototype implementation in the context
  of regular language separation and RMC.

\end{abstract}

\section{Introduction}
\label{sec:intro}

A central challenge in the formal verification of distributed systems is that
the number of participating processes is often not known \emph{a priori}.
In particular, \defn{parameterized} model checking~\cite{GS92} addresses the
correctness of a system where the number of copies of a process is a
\defn{parameter}.
A successful approach in this context—known as
\defn{regular model checking}~\cite{Abdulla12} (RMC)—consists
in encoding individual configurations as finite words,
sets of configurations as regular languages, and transitions as transducers.
For example, safety RMC consists in determining whether, from a regular set of
\defn{initial} configurations $S_0$, one can guarantee 
that the system is \defn{safe} and never reaches a regular set of \defn{bad} 
configurations $S_b$. By writing the successor (resp. predecessor)
function $\post$ (resp. $\pre$) and its transitive closure
$\post^*$ (resp. $\pre^*$) , this amounts to checking whether
$\post^*(S_0) \cap \pre^*(S_b) = \emptyset$.
RMC is, in general, undecidable; nevertheless, if one can find a language $L$
and effectively verify
that it contains only reachable states and no bad states,
then the system is safe.

\begin{example}
  \label{ex:equidist}
  Two tokens are initially located at both extremities of a track of even
  length
  $2 \times n$;
  $n \in \integers \setminus \{0\}$ is a \defn{parameter} of the system.
  Both tokens can only move synchronously in opposite directions.
  We want to prove that a collision can't happen,
  i.e. that the two tokens can't end up in the same position.
  
  A configuration is a word on the alphabet
  $\Sigma = \{\ton, \toff\}$ with exactly two $\ton$ letters, where 
  $\ton$ represents a position occupied by a token.
  The system's starting configuration is $\ton (\toff \toff)^{n-1} \ton$.
  Thus, the set of initial configurations for every possible value of $n$ 
  is $S_0 = \lang(\ton (\toff \toff)^* \ton)$.
  The set of forbidden configurations is
  $S_b = \lang(\toff^* \ton \toff \ton \toff^*)$ (the two tokens can overlap
  next step).
  
  The transition relation $\trans$ of the system is such that
  the tokens can only move toward or away from each other.
  As an example,
  $\toff\ton\toff\toff\ton\toff \trans \toff\toff\ton\ton\toff\toff$ and
  $\toff\ton\toff\toff\ton\toff \trans \ton\toff\toff\toff\toff\ton$.
  Thus, if we consider the set $\post^*$ of successors induced by $\trans$, $\{\toff\toff\ton\ton\toff\toff, \ton\toff\toff\toff\toff\ton\} 
  \subseteq \post^*(\{\toff\ton\toff\toff\ton\toff\})$.
\end{example}

To synthesize such a language, one may resort to
\defn{active automata learning}~\cite{Angluin87}
under a \defn{Minimally Adequate Teacher} (MAT):
a \defn{learner} infers a formal model (here, a 
regular language $L$ over a finite alphabet $\Sigma$) by asking membership—does 
$w \in \Sigma$ belong to $L?$—and equivalence—given a language
$H$, does $H = L$, and if it does not, what is a counterexample?—queries to a 
\defn{teacher}.
However, in some cases, it may happen that the learner is actually interested 
in learning \emph{any} language in a \defn{target} class of languages.

Consider the \defn{regular language separation} problem:
finding a regular language $L$—known as a \defn{separator}—such that $L_1
\subseteq L$ and $L \subseteq \overline{L_2}$ for some disjoint languages
$L_1$ and $L_2$.
Multiple solutions may exist:
words in $L_1$ (resp. $L_2$) should (resp. should not) belong to $L$, but words
that belong to neither are under no such constraint.
Yet, dedicated AAL techniques can yield a separator under the
\defn{inexperienced}~\cite{Leucker12} or \defn{incomplete}~\cite{MWSKF23}
teacher framework iMAT that allows ``don't know'' answers to membership queries
if a word's membership is unspecified;
it also replaces equivalence queries with \defn{validity} queries testing
separation.

RMC is a variation of this problem;
however, one may not be able to directly decide whether $L$ separates 
$\post^*(S_0)$ and $\pre^*(S_b)$,
e.g. if we cannot compute these two sets in the first place.
But another similar, sound approach known as
\defn{regular inductive invariant synthesis} can also guarantee safety:
finding a regular language $I$ called an \defn{inductive invariant}
such that $S_0 \subseteq I$, $I \subseteq \overline{S_b}$, and
$\post(I) \subseteq I$.
The language $I$ then obviously separates $\post^*(S_0)$ and $\pre^*(S_b)$.
This approach remains \defn{incomplete}:
a safe system may not admit an invariant.

\begin{example}
  \label{ex:equidist_invar}
  Consider Example~\ref{ex:equidist}.
  The language $I = \lang(\toff^* \ton (\Sigma^2)^* \ton \toff^*)$ is a regular 
  inductive invariant.
  The system is therefore safe for any value of $n$.
\end{example}

Were we to learn regular inductive invariants, we can easily check whether a 
hypothesis $H$ submitted to the teacher is inductive, i.e.
$\post(H) \subseteq H$;
if it is not, an \defn{inductive counterexample} such that
$(c_1 \in H) \land (c_2 \not\in H)$ yet $c_2 \in \post^*(\{c_1\})$
can be found.
Similarly, if $w_2 \in \post^*(\{w_1\})$ for some words $w_1$ and $w_2$, then 
any regular inductive invariant $I$ verifies $w_1 \in I \rightarrow w_2 \in I$;
the learner can make use of the \defn{inductive pair} $w_1 \rightarrow w_2$.

\paragraph{Our contributions.}
We introduce a new AAL framework IdMAT for non-uniquely specified targets that 
allows both inductive pairs and counterexamples (Sec.~\ref{sec:matandidmat}).
We pair this framework with $\linda$, a SAT-based AAL algorithm that extends 
Angluin's $\lstar$ algorithm~\cite{Angluin87}, making use of incremental SAT 
solving and UNSAT core analysis (Sec.~\ref{sec:lice}).
Inspired by Rivest and Schapire~\cite{RS93}, we design a frugal analysis method 
for simple and inductive counterexamples 
(Sec.~\ref{sec:ce_handling} and \ref{sec:ice_handling}).
Finally, we discuss our Python implementation of this framework, then its 
application to regular language separation and RMC (Sec.~\ref{sec:applications} 
and~\ref{sec:exp}).
When omitted or sketched, full proofs can be found in the technical report~\cite{IdMATTR}.

\subsection{Related Work}
\label{sec:related}

Chen et al.~\cite{Chen09}'s algorithm $L^\text{Sep}$ aims to learn a separator 
to two \emph{regular} languages by treating ``don't know'' answers as 
if they were the third letter $\square$ of an alphabet already containing $1$ 
and $0$, then relying on Angluin's $\lstar$ algorithm~\cite{Angluin87} to infer 
an incompletely specified finite state machine—also known as a 3DFA.
The latter's minimal \defn{specification} yields the minimal regular separator.
Yaacov et al.~\cite{DBLP:conf/models/YaacovWAH25} use a similar approach to 
learn concise bug descriptions.

Moeller et al.~\cite{MWSKF23} focus on the generic regular language separation 
problem.
Their $L^*_\square$ AAL algorithm maintains a queue of \defn{incomplete 
observation tables} (IOT, as defined in Sec.~\ref{sec:lice}), of which they 
fill in the blanks with a SAT solver to generate hypotheses.
Whenever the current IOT results in an UNSAT instance, its UNSAT
core~\cite{LiffitonS08} is used to insert one or more IOTs, one new prefix at a 
time.

Grinchtein et al.~\cite{Grinchtein06} synthesize minimal network invariants.
They allow inductive counterexamples of the form ``$w \in \post(I)$ but $w 
\not\in I$'' that they handle by maintaining an IOT queue and forking the 
current learner instance—one fork accepts the counterexample, the other rejects 
it.

As RMC amounts to determining whether the reachable fragment 
$\post^*(S_0)$ contains a bad configuration, many techniques have been
developed in order to compute this language, be it 
abstraction~\cite{AbdullaHH16}, widening~\cite{Touili01}, 
acceleration~\cite{Jonsson00} or AAL~\cite{AKMG04,Neider14,NeiderJ13}.
Other recent techniques combine abstraction with learning, and look for
an invariant of a given specific shape~\cite{EsparzaRW25,czerner2024}.

Chen et al.~\cite{ChenHLR17} use $\lstar$ with a \defn{generous} teacher
accepting \emph{any} regular inductive invariant. Instead of ``don't know'',
this teacher is \defn{strict} (resp. \defn{non-strict}) as it returns
$0$ (resp. $1$), thus targeting 
the smallest (resp. biggest) invariant.
It may be that neither are regular, despite a regular invariant existing.

Neider et al.~\cite{NeiderJ13,Neider14} compute regular inductive 
invariants with an $\lstar$-like AAL algorithm that allows ``don't know'' 
answers.
They maintain a single IOT instead of a queue but rely on an 
iterative SAT-based synthesis algorithm to generate a hypothesis
by treating the IOT as a passive learning sample.

The IdMAT framework formalizes these various use cases by framing 
incomplete learning as looking for an unspecified language in a
\defn{target class}.
A lack of consensus within the target class—a word is accepted by some targets 
but rejected by others—results in ``don't know'' answers to membership queries,
although IdMAT may—subsuming iMAT~\cite{MWSKF23} when it does not—nevertheless 
return inductive constraints verified by every language in the target class.

Moreover, among the $\lstar$ variants tailored for incomplete teachers, our AAL 
algorithm $\linda$ is the only one that doesn't maintain a learner queue or 
synthesize hypotheses with the help of a passive learning algorithm:
instead, it takes a single call to a SAT solver to generate a hypothesis from 
$\linda$'s only IOT.
However, $\linda$ does not necessarily return the smallest automaton in the 
target class:
this property is of limited use for RMC, as any inductive invariant proves 
safety.
Finally, rather than indiscriminately adding the entire set of 
prefixes~\cite{Grinchtein06,NeiderJ13,Neider14} or suffixes~\cite{MWSKF23} of a 
counterexample to the table, we extend Rivest-Schapire's algorithm~\cite{RS93} 
to incomplete teachers, for simple and inductive counterexamples alike.

\section{Preliminaries}
\label{sec:prelim}

\subsection{Regular languages}
\label{sec:ratlang}

An \defn{alphabet} is a finite set $\Sigma$.
A \defn{word} $w$ over $\Sigma$ is any finite sequence over $\Sigma$.
Its \defn{length} is written $|w|$.
A \defn{language} $L$ is a subset of the set $\Sigma^*$ of words.
Given $w_1, w_2 \in \Sigma^*$ and $L_1, L_2 \subseteq \Sigma^*$,
$w_1 \cdot w_2$ denotes their \defn{concatenation}—symbol $\cdot$ may be 
omitted—and
$L_1 \cdot L_2 = \{w_1 \cdot w_2 \mid w_1\in L_1, w_2 \in L_2\}$.
For any $i \in \integers$,
$w[i]$ stands for the $i$-th letter of $w$ if $0 \leq i < |w|$,
and $w[i] = \varepsilon$ otherwise.
For $0 < i \leq |w|$, $w^i$ is the \defn{prefix}
$w[0] \cdots w[i-1]$.
For $0 \leq i < |w|$, $\suffx{w}{i}$ is the \defn{suffix}
$w[i] \cdots w[|w|-1]$.
Both $w^i$ and $\suffx{w}{i}$ are equal to $\varepsilon$ for other values of 
$i$.

\begin{definition}
  A \defn{Deterministic Finite Automaton} (DFA)
  $\aut = (Q, \Sigma, \delta, q_i, F)$ is a tuple s.t.
  $Q$ is a finite \defn{set of states}, $\Sigma$ a finite alphabet,
  $\delta: Q \times \Sigma \rightarrow Q$ a \defn{transition function},
  $q_i \in Q$ the \defn{initial} state,
  $F \subseteq Q$ the set of \defn{final} states.
\end{definition}

The inductive closure $\delta^*: Q \times \Sigma^* \to Q$ of $\delta$ is s.t.
for $w \in \Sigma \setminus \{\varepsilon\}$ and $q \in Q$,
$\delta^*(q, \varepsilon) = q$ and
$\delta^*(q, w) = \delta(q', \suffx{w}{1})$ where $q' = \delta^*(q, w[0])$.
If $\delta^*(q_i, w) = q$, $w$ is an \defn{access sequence} for state $q$.
$\aut$ \defn{accepts} $w$ if $\delta^*(q_i, w) \in F$;
we denote this predicate $\aut(w)$.
The language $\lang(\aut)$ \defn{accepted} by $\aut$ is the set of all words
accepted by $\aut$.
A language is \defn{regular} or \defn{rational} if it is accepted by some DFA.
$\Rat(\Sigma)$ denotes the set of regular languages over $\Sigma$,
characterized as follows:

\begin{definition}[Congruence relation]
  Given $L \subseteq \Sigma^*$, $\lequiv$ is the equivalence relation on
  $\Sigma^*$ such that $w_1 \lequiv w_2$ if and only if
  $\forall s \in \Sigma^*$,
  $w_1 \cdot s \in L \leftrightarrow w_2 \cdot s \in L$.
\end{definition}

\begin{theorem}[Myhill-Nerode]
  \label{th:mn}
  $L$ is regular if and only if $\lequiv$ has finite index.
\end{theorem}

From $\lequiv$'s finite quotient space $\lqspace$, one can infer the unique 
minimal (in terms of states) DFA
$\aut_L = (\lqspace, \Sigma, \delta_L, [\varepsilon]_{\lequiv}, F_L)$ 
accepting $L$, where $\delta_L([p]_{\lequiv}, a) = [pa]_{\lequiv}$ and
$F_L = \{[w]_{\lequiv} \mid w \in L\}$.
It is called $L$'s \defn{canonical} DFA.

If $w_1 \cdot s \in L \notleftrightarrow w_2 \cdot s \in L$ for some
$s \in \Sigma^*$, we say that $s$ \defn{distinguishes} $w_1$ and $w_2$;
it is witness to $w_1 \notlequiv w_2$.
A \defn{state cover} of $L$ is a set $P$ such that for any congruence class 
$[w]_{\lequiv} \in \lqspace$, there exists a unique $p \in P$,
$p \in [w]_{\lequiv}$.
A set $S$ is a \defn{characterization set} of $L$ if for any
$w_1, w_2 \in \Sigma^*$ such that $w_1 \notlequiv w_2$,
there exists $s \in S$ distinguishing $w_1$ and $w_2$.

\subsection{Active automata learning}
\label{sec:aal}

We recommend Steffen et al.'s chapter~\cite{Steffen2011} for a more thorough 
introduction to active automata learning.

\paragraph{The MAT framework.}
Let $L \in \Rat(\Sigma)$ be a target language.
The \defn{Minimally Adequate Teacher} (MAT) for $L$ consists of two functions:
a \defn{membership} oracle $\MEM: w \in \Sigma^* \mapsto ``w \in L?" \in 
\{0, 1\}$ and an \defn{equivalence} oracle $\EQ: \Rat(\Sigma) \to \Sigma^* 
\cup \{1\}$ such that $\EQ(H) = 1$ if $H = L$ and $\EQ(H) = w$ s.t. $w \in H 
\notleftrightarrow w \in L$ otherwise—$w$ is said to be a 
\defn{counterexample}.
\defn{Active Automata Learning} (AAL) is a two player game consisting in a 
\defn{learner} using the MAT to compute $L$, trying to perform as few
calls to the oracles $\MEM$ and $\EQ$—respectively called membership queries 
(MQs) and equivalence queries (EQs)—as possible.

\paragraph{Data structures.}
Angluin's $\lstar$ algorithm~\cite{Angluin87} maintains two sets $P, S 
\subseteq \Sigma^*$, informally called the sets of \defn{prefixes} and 
\defn{suffixes}.
They determine the MQs to be submitted to the MAT, the answers to which being 
stored in a table:

\begin{definition}
  An \defn{Observation Table} (OT) is a function
  $\obstab: (P \cup (P \cdot \Sigma)) \times S \to \{0, 1\}$
  such that $\obstab(p, s) = \MEM(p \cdot s)$
  (also written $\obstab(p \cdot s) = \MEM(p \cdot s)$).
\end{definition}

The set $P \cdot \Sigma$ is called the \defn{frontier}, and
$\cells = (P \cup (P \cdot \Sigma)) \cdot S$, the set of \defn{cells} of 
$\obstab$.
The latter is the set of words one needs to query to fill in the table.

An OT $\obstab$ induces a congruence relation $\tequiv$ over $P \cup (P 
\cdot \Sigma)$ such that $w_1 \tequiv w_2$ if and only if $\forall s \in 
S$, $\MEM(w_1 \cdot s) = \MEM(w_2 \cdot s)$, i.e. lines $w_1$ and $w_2$ of 
$\obstab$ are 
identical.
Note that $\tequiv$ under-approximates $\lequiv$, as
$w_1 \nottequiv w_2 \rightarrow w_1 \notlequiv w_2$.

An OT is \defn{sharp} if $\forall p_1, p_2 \in P$ such that $p_1 \neq p_2$, 
$p_1 \nottequiv p_2$:
its prefixes are demonstrably pairwise distinguished.
It is \defn{closed} if $\forall w \in P \cdot \Sigma$,
$\exists p \in P$, $p \tequiv w$:
any line labeled by a word $w \in P \cdot \Sigma$ in the frontier
is identical to a line labeled by a prefix $p$.
If $\obstab$ is sharp and closed, this prefix $p$ is unique, called the
\defn{representative} of $w$, and written $\repr{w}{\obstab}$.
A sharp and closed OT induces a \defn{hypothesis}
$\hyp_\obstab = (P, \Sigma, \delta_\obstab, \varepsilon, F_\obstab)$ 
such that $F_\obstab = \{p \in P \mid \MEM(p) = 1\}$ and $\forall p \in P$, 
$\forall a \in \Sigma$, $\delta_\obstab(p, a) = \repr{p \cdot a}{\obstab}$.

For any $w \in \Sigma^*$, we define
$\repr{w}{\obstab} = \delta^*_\obstab(p, w)$.
This does not contradict the previous definition of $\repr{w}{\obstab}$ on $P \cdot \Sigma$ by design of $\delta_\obstab$.
Moreover, $\hyp_\obstab(\repr{w}{\obstab}) = \hyp_\obstab(w)$
for any $w \in \Sigma^*$ and 
$\repr{p}{\obstab} = p$ for any $p \in P$.
Intuitively, $\repr{w}{\obstab}$ is the prefix in $P$ identifying the state $w$ leads to in $\hyp_\obstab$.

\paragraph{The learning algorithm $\lstar$.}
Initially, $P = S = \{\varepsilon\}$.
Algorithm $\lstar$ queries the MAT as it computes $\obstab$, then
seeks closure defects, i.e. whether $\exists p \cdot a \in P \cdot \Sigma$
such that $\forall p' \in P$, $p \cdot a \not\equiv_\obstab p'$.
If one such $p \cdot a$ exists, it is added to $P$.
The OT is therefore always kept sharp by design.
Once it is closed, the learner computes $\hyp_\obstab$.
If $\EQ(\hyp_\obstab)$ returns a counterexample $w$, then it can be proven that 
there exists a \defn{refining} suffix of $w$ s.t. its addition to $S$ causes a 
closure defect.

The sharpness of $\obstab$ guarantees termination due to the elements of $P$ 
demonstrably belonging to different classes of the quotient space $\lqspace$.
$\lstar$'s correction is a consequence of $\tequiv$ being an 
under-approximation of $\lequiv$.
As $\lstar$ ends, $P$ is a state cover of $L$ and $S$ a characterization set.

\paragraph{Counterexample analysis.}
A method designed by Rivest and Schapire~\cite{RS93} (RS) infers
a refining suffix $s$ of a counterexample $w$ in $\bigo(\log|w|)$ MQs.
Intuitively, it consists in replacing a prefix of $w$ with its representative
and observing at which point doing so alters $\MEM$'s output.
Formally, for $i \in \interv{0}{|w|}$, we consider the $i$-th
\defn{partial computation} $\pcomp[\obstab]{w}{i} = \repr{w^i}{\obstab} \cdot 
\suffx{w}{i}$—we may omit $\obstab$ whenever contextually relevant—and its 
matching \defn{evaluation predicate} $\agree{w}{i} = \MEM(\pcomp{w}{i})$.

\begin{definition}
  A \defn{breaking point} (BP) w.r.t. counterexample $w \in \Sigma^*$ to 
  $\hyp_\obstab$ is $i \in \interv{0}{|w|-1}$ s.t.
  $\agree{w}{i} \neq \agree{w}{i+1}$.
\end{definition}

If $i$ is a BP, $\suffx{w}{i}$ separates $\repr{w^i}{} \cdot 
w[i]$ and $\repr{w^{i+1}}{}$, yet $\repr{w^i}{} \cdot w[i] \tequiv 
\repr{w^{i+1}}{}$.
Thus, adding $\suffx{w}{i}$ to $S$ causes a closure defect and induces a 
refinement of $\hyp$.
A BP always exists due to $\agree{w}{0} \neq \agree{w}{|w|}$.
If we explore table $[\agree{i}{w}]_{i \in \interv{0}{|w|}}$ dichotomically,
finding a BP only requires a logarithmic number of MQs.

\subsection{SAT solving}

Given a set of \defn{variables} $\vars$, we consider the set
$\varphi \vcentcolon= 0 \mid 1 \mid x \in \vars \mid \neg \varphi
\mid \varphi \land \varphi \mid \varphi \lor \varphi
\mid \varphi \rightarrow \varphi \mid \varphi \leftrightarrow \varphi$
of \defn{propositional formulas} (PFs) $\propfor$.
A \defn{valuation} $v$ is a partial function $v: \vars \to \{0, 1\}$.
A valuation $v'$ \defn{subsumes} $v$ if $v$ is a subset of $v'$ in terms of 
binary relations over $\vars \times \{0, 1\}$.
The PF $v[\varphi]$ is obtained by replacing any variable $x \in \vars$ occuring in
$\varphi$ by $v(x)$ whenever the latter is defined.

A \defn{model} $m$ for a formula $\varphi$—written $m \vDash \varphi$—is a 
valuation s.t. $m[\varphi] = 1$ according to Boolean semantics 
and every variable of $\varphi$ belongs to $m$'s domain.
The \defn{satisfiability} (SAT) problem for a PF $\varphi$
consists in determining whether a model for $\varphi$ exists;
$\varphi$ is said to be SAT if it does, and UNSAT otherwise.

A \defn{set of clauses} is a finite set $\clauses$ of PFs in 
conjunctive normal form (CNF), interpreted as the PF
$\underset{c \in \clauses}{\bigwedge} {c}$.
If $\clauses$ is UNSAT, an \defn{UNSAT core} of $\clauses$ is a minimal
(according to $\subseteq$) subset $\clauses'$ of $\clauses$ such that
$\clauses'$ is UNSAT as well.
A \defn{SAT solver} is a program that decides for any set of
clauses $\clauses$ whether $\clauses$ is SAT or not and returns a model $m$ of 
$\clauses$ if it is.
A solver is \defn{incremental} if we can add new clauses to $\clauses$ and try
to solve the updated problem without resetting the solver.

\section{The \idmat{} Framework}
\label{sec:matandidmat}

\subsection{Inductive oracles}
\label{sec:def_idmat}

We introduce an extension to the MAT framework where the target language is 
not unique:
given a \defn{target class} $\tgt \subseteq 2^{\Sigma^*}$,
the goal of the learner in this setting is to identify \emph{one}
unspecified language $T \in \tgt$.
To this end, we introduce an inductive teacher IdMAT consisting of two 
oracles:

\begin{definition}
  \label{def:memindoracle}
  An \defn{inductive membership oracle} is a function $\MEM: \Sigma^* \times 
  2^{\Sigma^*} \to \{0, 1\} \cup 2^{\Sigma^* \times \Sigma^*}$
  such that for any finite set $A \subseteq \Sigma^*$ and $w \in \Sigma^*$,
  exactly one of the following holds:
  \begin{itemize}
    \item $\MEM(w, A) = 1$ and $\forall T \in \tgt, w \in T$;
    every target language accepts $w$.
    \item $\MEM(w, A) = 0$ and $\forall T \in \tgt, w \notin T$;
    every target language rejects $w$.
    \item $\MEM(w, A) \subseteq 
    \{(w_1, w_2) \in (A \times \{w\}) \cup (\{w\} \times A) \mid
    \forall T \in \tgt, w_1 \in T \rightarrow w_2 \in T\}$.
  \end{itemize}
\end{definition}  

Intuitively, the third case of Definition~\ref{def:memindoracle} is an 
\defn{incomplete} answer that states that $w$'s membership status is not constant 
over the target class.
Nevertheless, the learner can provide ``hint words'' in a set $A$, and
the teacher can then provide \defn{inductive pairs} $(w_1,w_2)$ relating $w\in \{w_1,w_2\}$
to these words,
in the form
``For any target $T$, if $w_1$ belongs to $T$, then so must $w_2$.''.
As a shorthand, we write $\MEM(w) = \MEM(w, \emptyset)$ reducing the
possible answers to $0$, $1$ or $\emptyset$, that is,
respectively ``Yes'', ``No'' or ``Don't know'' answers.

Similarly, the following Definition~\ref{def:eqindoracle} generalizes equivalence
queries to account for inductive pairs being counterexamples.
Moreover, notice that
while the oracles are not-uniquely defined, the provided
counterexamples must be consistent with the answers returned by the
membership oracle:

\begin{definition}
  \label{def:eqindoracle}
  An \defn{inductive validity oracle} is a function
  $\VAL: \Rat(\Sigma) \to \{1\} \cup \Sigma^* \cup (\Sigma^* \times \Sigma^*)$ 
  such that for any hypothesis $H \in \Rat(\Sigma)$:
  \begin{itemize}
    \item $\VAL(H) = 1$ if and only if $H \in \tgt$;
    the hypothesis belongs to the target class.

    \item If $\VAL(H) = w \in \Sigma^*$,
        then $\MEM(w) \in \{0,1\}$ and $w \in H
        \leftrightarrow \MEM(w) = 0$; word $w$ is a \defn{simple} counterexample.
    \item If $\VAL(H) = (w_1, w_2) 
    \in \Sigma^* \times \Sigma^*$, then
    $(w_1, w_2) \in \MEM(w_1, \{w_2\})$, $w_1 \in H$ but
    $w_2 \notin H$; pair $(w_1, w_2)$ is an \defn{inductive} counterexample.
  \end{itemize}
\end{definition}

\subsection{Instantiating IdMAT}
\label{sec:applications}

\paragraph{Regular language separation.}
Consider two languages $L_1, L_2 \subseteq \Sigma^*$.
Let $\tgtsep = \{L \in \Rat(\Sigma) \mid L_1 \subseteq L \subseteq 
\overline{L_2}\}$ be the target class of regular separators.
For any $w \in \Sigma^*$, finite $A \subseteq \Sigma^*$,
and $H \in \Rat(\Sigma)$, IdMAT oracles complying with 
Definitions~\ref{def:memindoracle} and 
\ref{def:eqindoracle} are defined as follows:
\begin{itemize}
  \item $\MEM(w, A) = 1$ if $w \in L_1$, $0$ if $w \in L_2$,
  and $\emptyset$ otherwise.
  \item $\VAL(H)$ is any word in
  $(H \setminus L_2) \cup (L_1 \setminus H)$ if $H \not\in \tgtsep$, and
  $1$ otherwise.
\end{itemize}

They are effective if $w \in L_1$,
$w \in L_2$, $L_1 \subseteq L$ and $L \subseteq \overline{L_2}$ are 
decidable for every $w \in \Sigma^*$ and $L \in \Rat(\Sigma)$.
The inductive capabilities (hint words, inductive counterexamples) of the 
IdMAT are ignored,
restricting $\MEM$'s output to $0$, $1$ and the ``don't know'' answer 
$\emptyset$,
thus subsuming Moeller et al.'s iMAT framework~\cite{MWSKF23}.

\paragraph{Regular model checking.}
Let $\rts = (\Sigma, {\trans}, S_0, S_b)$ be a
\defn{regular transition system} (RTS) with $S_0, S_b \in \Rat(\Sigma)$ and
${\trans} \in \Rat(\Sigma \times \Sigma)$ a length-preserving finite transducer 
inducing an effectively computable \defn{successor} relation
$\post: \Rat(\Sigma) \to \Rat(\Sigma)$ such that
$\post(L) = \{w_2 \mid w_1 \in L, (w_1, w_2) \in {\trans}\}$.
Similarly, $\pre: \Rat(\Sigma) \to \Rat(\Sigma)$ must be an effectively 
computable \defn{predecessor} relation such that
$\pre(L) = \{w_1 \mid w_2 \in L, (w_1, w_2) \in {\trans}\}$.
However, note that
$\post^*(L) = \underset{k \geq 0}{\bigcup} \post^k(L)$ and  
$\pre^*(L) = \underset{k \geq 0}{\bigcup} \pre^k(L)$ may not be computable.

Consider the class
$\tgtind = \{I \subseteq \Sigma^* \mid (S_0 \subseteq I) \land
(I \subseteq \overline{S_b}) \land (\post(I) \subseteq I)\}$
of inductive invariants.
For every $w \in \Sigma^*$, finite $A \subseteq \Sigma^*$,
and $H\in \Rat(\Sigma)$, consider the IdMAT oracles complying with 
Definitions~\ref{def:memindoracle} and \ref{def:eqindoracle}:
\begin{itemize}
  \item $\MEM(w, A) = 1$ if $w \in \post^*(S_0 \cap \Sigma^{|w|})$,
  $0$ if $w \in \pre^*(S_b \cap \Sigma^{|w|})$,
  and    
  $\{w\} \times (\post^*(\{w\})\cap A) \cup (\pre^*(\{w\})\cap A)\times \{w\}$
  otherwise.
  \item $\VAL(H) = w$ for some $w \in (S_0 \setminus H) \cup (S_b \cap H)$
  if such a word exists,
  else $(w_1, w_2)$ for some $w_1 \in H$ and $w_2 \notin H$
  such that $(w_1, w_2) \in \MEM(w_1, \{w_2\})$ if such a pair exists,
  else $1$.
\end{itemize}

These oracles are effectively computable:
$\trans$ is length-preserving and $S_0 \cap \Sigma^{|w|}$ is finite, thus
$\post^*(S_0 \cap \Sigma^{|w|})$ is finite and can be found using an 
iterative fixpoint algorithm.
The same is true of $\pre^*(S_b \cap \Sigma^{|w|})$, $\post^*(\{w\})\cap A$, 
and $\pre^*(\{w\})\cap A$.
Moreover, $\post(H)$ is regular and we can look for inductive counterexamples 
by seeking $w_2 \in \post(H) \setminus H$ then computing
$w_1 \in \pre^*(\{w_2\}) \cap H$.

\begin{example}
  \label{exple:equidist_indpair}
  Consider Example~\ref{ex:equidist}.
  Words $w_1 = \ton\toff\toff\ton\toff$ and $w_2 = \toff\ton\ton\toff\toff$   
  belong neither to $\post^*(S_0)$ nor $\pre^*(S_b)$, being configurations of 
  odd length with an even number of $\toff$ empty positions in-between their 
  two $\ton$ tokens.
  Yet $w_2 \in \post^*(\{w_1\})$, thus $(w_1, w_2)$ is an inductive pair
  and $(w_1, w_2) \in \MEM(w_1, \{w_2\})$.
\end{example}

\section{The $\linda$ Active Learning Algorithm}
\label{sec:lice}

We assume that the oracles defined in Section~\ref{sec:def_idmat} are 
effectively computable for a non-empty target class $\tgt$.
We aim to find one $T \in \tgt \cap \Rat(\Sigma)$ if it exists and to this end 
introduce $\linda$, a new AAL algorithm under the IdMAT framework.

\subsection{Data structures}
\label{sec:data}

Algorithm $\linda$ is based on $\lstar$.
Similarly, it maintains two sets $P$ and $S$ and an 
OT variant $\obstab$ that also accounts for incomplete answers:

\begin{definition}
  An \defn{Incomplete Observation Table} (IOT) is
  $\obstab: (P \cup (P \cdot \Sigma)) \times S \to \{0, 1, \square\}$ s.t.
  $\obstab(p, s) = \square$ if $\MEM(p \cdot s) \not\in \{0, 1\}$
  and $\MEM(p \cdot s)$ otherwise.
\end{definition}

Symbol $\square$ stands for incomplete answers to MQs, in a similar fashion to 
$\lstar_\square$~\cite{MWSKF23}.
An IOT is shown in Example~\ref{ex:equidist_unsat}.
A sub-table $\obstab_B$ of $\obstab$ for some $B \subseteq P$ is $\obstab$'s 
restriction to $(B \cup (B \cdot \Sigma)) \times S$.
Unlike $L^*$ under the MAT framework or $\lstar_\square$ under iMAT, we must 
also account for inductive answers to MQs that constrain $\obstab$'s $\square$ 
cells,
since $\obstab$ does not fully record these:

\begin{definition}
  \label{def:indpairs}
 Let $U = \{w \in \cells \mid \obstab(w) = \square\}$ be the set of 
 \defn{incomplete} cells.
 The set of \defn{inductive pairs} is
 $\indpairs = \underset{w \in U}{\bigcup} \MEM(w, U)$.
\end{definition}

\subsection{Outlining the learning loop}
\label{sec:main_loop}

Unlike other AAL algorithms for incomplete teachers, $\linda$ does not 
maintain several IOTs nor use passive learning algorithms to infer a 
hypothesis from the table.
Instead, our intuition is to maintain a single IOT $\obstab$,
but design $P$ and $S$ in such a fashion
they will at some point respectively \emph{contain}
a state cover $B$—known as a \defn{basis}—and a characterization set of some
$T \in \tgt \cap \Rat(\Sigma)$.

To infer a hypothesis $H$ that we can submit to the IdMAT, a 
prefix-closed subset $B \subseteq P$ called the \defn{basis} is chosen, 
inducing a sub-table $\obstab_B$.
The ``blanks'' of $\obstab_B$ are then ``filled''—i.e. $\square$ cells are 
arbitrarily assigned values—in a manner that makes the resulting table 
sharp, closed, and compliant with the inductive constraints recorded in 
$\indpairs$.
We encode the existence of such a basis and sub-table as a SAT instance 
$\clauses_{P, S}$;
given a model $m$ to this instance, $\obstab_m$ denotes the OT that results 
from using $m$ to choose a basis $B$ and fill IOT $\obstab_B$'s $\square$ cells.

Algorithm~\ref{algo:lice} summarizes $\linda$'s main loop.
Similarly to $\lstar$, $P, S = \{\varepsilon\}$ initially.  
In a broad sense, closure defects result in $P$ being updated while 
counterexamples yield new suffixes to add to $S$, although the finer details 
of these steps significantly differ from $L^*$ under the MAT framework.
Note that if the target class is a singleton $\tgt \cap \Rat(\Sigma) = \{L\}$, 
$\MEM$ never returns incomplete answers, and a run of $\linda$ is essentially 
identical to a run of $L^*$ with RS under the MAT framework:
it maintains the same OT and performs the same queries.

\begin{algorithm}
  \begin{algorithmic}
    \State $P, S \gets \{\varepsilon\}$
    \Loop
      \State Update the IOT $\obstab$ and the clauses
      $\clauses_{P, S}$.
      \Comment{See Section~\ref{sec:fill_subtable} and 
        \appref{A}{app:sat_clauses}.
        }
      \While{$\clauses_{P, S}$ is satisfied by some model $m$}
        \State From $m$, build a hypothesis $\hypm$.
        \Comment{See Section~\ref{sec:fill_subtable} and 
        \appref{A}{app:sat_clauses}.
        }
        \State $c \gets \VAL(\hypm)$
        \If{$c = 1$}
          \Return $\hypm$
          \Comment{An invariant has been found.}
        \ElsIf{$c = w \in \Sigma^*$}
          \State Analyze $w$ and update $S$, $\obstab$ and $\clauses_{P, S}$.
          \Comment{See Section~\ref{sec:ce_handling}.}
        \ElsIf{$c = (w_1, w_2)$}
          \State Analyze $(w_1, w_2)$ and update $S$, $\obstab$ and 
          $\clauses_{P, S}$.
          \Comment{See Section~\ref{sec:ice_handling}.}
        \EndIf
      \EndWhile
      \State $P \gets P \cup \{{p \cdot a} \mid \cclos{p, a} \in
      \texttt{UnsatCore}(\clauses_{P, S})\}$
      \Comment{See Section~\ref{sec:fill_subtable}.}
    \EndLoop
  \end{algorithmic}
  \caption{The learning algorithm \linda.}
  \label{algo:lice}
\end{algorithm}

\subsection{Generating the hypothesis}
\label{sec:fill_subtable}

We encode as a SAT instance the existence of both a basis $B$ and a membership 
assignment function that would make the resulting sub-table sharp and closed.
A set of clauses $\clauses_{P, S}$ detailed in \appref{A}{app:sat_clauses}
guarantees that, if $m$ is a model of $\clauses_{P, S}$, then
it yields a sub-table $\obstab_m$ of $\obstab$ that is sharp and closed.
For each $w \in \cells$ such that $\obstab(w) = \square$, we introduce a
variable $x_w$ representing $w$'s membership status.
For every $p \in P$, a variable $b_p$ determines whether $p$ belongs to $B$.
Finally, variable $e_{p, a, q}$ states that $p \cdot a \in P \cdot \Sigma$ 
must be equivalent to $q \in P$ w.r.t. the OT generated.
Of particular interest are the clauses $\cclos{p, a}$ for $p \in P$ and $a 
\in \Sigma$ that guarantee $p \cdot a \equiv_{\obstab_m} p'$ for some $p' \in 
P$, hence, $\obstab_m$'s closure:
\begin{equation}
  \tag{$\mathtt{clos}_{p, a}$}
  b_p \rightarrow\bigvee_{p'\in P} e_{p, a, p'}
\end{equation}
Moreover, for any pair $(w_1, w_2) \in \indpairs$, there exists a clause 
$x_{w_1} \rightarrow x_{w_2}$ in $\clauses_{P, S}$.

\paragraph{Using the UNSAT core to extend the set of prefixes.}  
Should the instance $\clauses_{P, S}$ end up being UNSAT,
then it can be proven (see \appref{A}{app:sat_clauses} for a detailed proof)
that its core 
always contains at least one clause $\cclos{p, a}$.
Then for each $(p,a) \in P \times \Sigma$ such that
$\cclos{p, a}$ is in the core, we add $p \cdot a$ to $P$.
Intuitively, an UNSAT instance means that, no matter the basis chosen and the 
values assigned to $\square$ cells, the resulting sub-table will contain a 
closure defect.
Theorem~\ref{thm:termination} later states that these core-based refinements
guarantee termination.

\begin{example}
  \label{ex:equidist_unsat}
  Consider the following IOT $\obstab$ induced by Example~\ref{ex:equidist} 
  after counterexample $\ton\toff\ton$ to hypothesis $H = \Sigma^*$
  results in a first refinement of $\obstab$.
  
  \begin{center}
    \begin{minipage}{0.58\linewidth}
      Due to $|P| = 1$, the only possible basis is $B =\{\varepsilon\}$ and 
      $\obstab_B = \obstab$.
      The closure condition results in $\varepsilon \equiv_{\obstab_m} \toff 
      \equiv_{\obstab_m} \ton$ for any possible model $m$.
      Thus, $\bleu[\obstab_m(\varepsilon, \toff\ton) = 0]$ and 
      $\rouge[\obstab_m(\toff, \ton) = 1]$.
      But both cells describe the membership status of word $\toff\ton$.
      There is a contradiction, the instance is UNSAT, and $\toff$, $\ton$ are 
      added to $P$.
    \end{minipage}
    \hfill
    \begin{minipage}{0.38\linewidth}
    {\renewcommand{\arraystretch}{1.1}
      \setlength{\tabcolsep}{6pt}
      \begin{tabular}{|c|c|ccc|}
        \hline
        \multicolumn{2}{|c|}{} & 
        \multicolumn{3}{c|}{\cellcolor[gray]{0.9}$\mathbf{S}$} \\
        \cline{3-5}
        \multicolumn{2}{|c|}{} & $\varepsilon$ & $\ton$ & $\toff\ton$ \\
        \hline
        \cellcolor[gray]{0.9}$\mathbf{P}$ & $\varepsilon$ & 
        $\square$ & $\square$ & $\bleu[\square]$ \\        
        \hline
        \cellcolor[gray]{0.9} & $\toff$ & 
        $\square$ & $\rouge[\square]$ & $\square$ \\
        \multirow{-2}{*}{\cellcolor[gray]{0.9}$\mathbf{P \cdot \Sigma}$} & 
        $\ton$ & $\square$ & $1$ & $0$ \\
        \hline
      \end{tabular}}
    \end{minipage}
  \end{center}
\end{example}

\paragraph{Building the hypothesis.}
If the instance is SAT, we denote $m$ the model returned by the SAT solver 
and $\obstab_m$ the resulting $\square$-free sub-table. 
From  $\obstab_m$, a hypothesis $\hypm$ can be built, as described in 
Section~\ref{sec:ratlang}.

\begin{example}
  \label{ex:equidist_hyp}
  Consider the updated IOT resulting from Example~\ref{ex:equidist_unsat}'s 
  refinement.
  
  \begin{center}
    \begin{minipage}{0.56\linewidth}
      Were the SAT solver to pick $\rouge[B = \{\varepsilon, \ton\}]$ as a 
      basis,
      grayed lines $\toff\toff$ and $\toff\ton$ are \gris[dropped] from 
      $\obstab_B$.
      Filling the remaining lines and $\square$ cells as depicted results in a 
      sharp, closed table $\obstab_m$ from which we can infer the following 
      two state DFA $\hypm$:
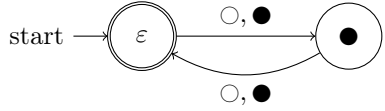
      \begin{center}
        \begin{tikzpicture}[x=2.75cm, y=2cm]
          \node[state, initial, accepting] (q0) at (0,0) {$\varepsilon$};
          \node[state] (q1) at (1,0) {$\ton$};
        
          \path (q0) edge[->] node[above] {$\toff, \ton$} (q1);
          \path (q1) edge[->, bend left] node[below] {$\toff, \ton$} (q0);
        \end{tikzpicture}
      \end{center}
    \end{minipage}
    \hfill
    \begin{minipage}{0.40\linewidth}
      {\renewcommand{\arraystretch}{1.1}
      \setlength{\tabcolsep}{6pt}
      \begin{tabular}{|c|c|ccc|}
        \hline
        \multicolumn{2}{|c|}{} & 
        \multicolumn{3}{c|}{\cellcolor[gray]{0.9}$\mathbf{S}$} \\
        \cline{3-5}
        \multicolumn{2}{|c|}{} & $\varepsilon$ & $\ton$ & $\toff\ton$ \\
        \hline
        \cellcolor[gray]{0.9} & $\rouge[\varepsilon]$ & 
        $\fbox{\scriptsize{1}}$ & $\fbox{\scriptsize{0}}$ & 
        $\fbox{\scriptsize{1}}$ \\  
        \cellcolor[gray]{0.9} & $\toff$ & 
        $\fbox{\scriptsize{0}}$ & $\fbox{\scriptsize{1}}$ & 
        $\fbox{\scriptsize{0}}$ \\
        \multirow{-3}{*}{\cellcolor[gray]{0.9}$\mathbf{P}$} & 
        $\rouge[\ton]$ & $\fbox{\scriptsize{0}}$ & $1$ & $0$ \\
        \hline
        \cellcolor[gray]{0.9} & $\gris[\toff\toff]$ & 
        $\gris[\square]$ & $\gris[\square]$ & $\gris[\square]$ \\
        \cellcolor[gray]{0.9} & $\gris[\toff\ton]$ & 
        $\gris[\square]$ & $\gris[\square]$ & $\gris[0]$ \\
        \cellcolor[gray]{0.9} & $\ton\toff$ & 
        $\fbox{\scriptsize{1}}$ & $0$ & $1$ \\
        \multirow{-4}{*}{\cellcolor[gray]{0.9}$\mathbf{P \cdot \Sigma}$}
        & $\ton\ton$ & $\fbox{\scriptsize{1}}$ & $\fbox{\scriptsize{0}}$ & 
        $\fbox{\scriptsize{1}}$ \\
        \hline
      \end{tabular}}
    \end{minipage}
  \end{center}
\end{example}

\subsection{Analyzing simple counterexamples}
\label{sec:ce_handling}
  
Under incomplete teachers, the evaluation predicate introduced in 
Section~\ref{sec:aal} may return incomplete answers.
Thus, the existence of a breaking point is no longer guaranteed—as shown in 
Example~\ref{ex:equidist_sce}—and incomplete learning algorithms revert to 
adding the entire set of prefixes~\cite{Grinchtein06,NeiderJ13,Neider14} 
(resp. suffixes~\cite{MWSKF23}) of a counterexample $w$ to $P$ (resp. $S$)
in order to induce a refinement.

We introduce a counterexample analysis algorithm that can reduce the 
number of suffixes added to $S$, in a fashion similar to RS under the MAT 
framework.
To this end, the evaluation predicate must account for incomplete answers and
membership assumptions made by model $m$ while assigning values to $\square$ 
cells:
  
\begin{definition}
  Given a counterexample $w \in \Sigma^*$ to $\hypm$ and
  $i \in \interv{0}{|w|}$, the \defn{incomplete evaluation predicate} is
  \begin{enumerate*}
    \item $\agree{w}{i} = \MEM(\pcomp[\obstab_m]{w}{i})$ if
    $\MEM(\pcomp[\obstab_m]{w}{i}) \in \{0, 1\}$, else
    \item $\agree{w}{i} = m[x_{\pcomp[\obstab_m]{w}{i}}]$ if 
    $m[x_{\pcomp[\obstab_m]{w}{i}}]$ is defined, else
    \item $\agree{w}{i} = \square$.
  \end{enumerate*}
\end{definition}

While a BP may not exist, some intervals can nevertheless provide $\linda$ with 
a set of new suffixes whose addition to $S$ induces a refinement of the 
hypothesis:
  
\begin{definition}
  A \defn{Breaking Interval} (BI) is an interval $\interv{i}{j}$,
  $i, j \in \interv{0}{|w|}$, $i < j$, such that 
  $\{\agree{w}{i}, \agree{w}{j}\} \subseteq \{0, 1\}$ and $\agree{w}{i} \neq 
  \agree{w}{j}$.
\end{definition}

\begin{property}
  \label{prop:bp_bounds}
  $\{\agree{w}{0}, \agree{w}{|w|}\} \subseteq \{0, 1\}$, 
  $\agree{w}{0} \neq \agree{w}{|w|}$, and $\interv{0}{|w|}$ is a BI. 
\end{property}
  
\begin{example}
  \label{ex:equidist_sce}
  Consider counterexample $w = \toff\ton\toff\toff\toff\ton \in \pre^*(S_b)$ 
  to hypothesis $\hypm$ of Example~\ref{ex:equidist_hyp}.
  Let us perform a RS analysis of $w$:
  we split $w$ into a prefix $w^i$ and a suffix $\suffx{w}{i}$,
  compute the representative $\repr{w^i}{\hypm}$ of $w^i$,
  explicit the partial computation $\pcomp[\obstab_m]{w}{i} = \repr{w^i}{\hypm} \cdot \suffx{w}{i}$, 
  and finally perfom a query $\MEM(\pcomp[\obstab_m]{w}{i})$ to compute $\agree{w}{i}$.
  It results in the following RS table that lacks a BP: 
  
  \begin{center}
  	\noindent
    \scalebox{0.83}{\renewcommand{\arraystretch}{1.1}
	    \setlength{\tabcolsep}{6pt}
	    \begin{tabular}{|c|c|c|c|c|c|c|c|}
	      \hline
	      \cellcolor[gray]{0.9}$i$  & $\bleu[0]$ & $\bleu[1]$ & $\bleu[2]$ & 
	      $\bleu[3]$ & $4$ & $5$ & $6$ \\
	      \hline
	      \cellcolor[gray]{0.9}$w^i$ & $\varepsilon$ & $\toff$ & 
	      $\toff\ton$ & $\toff\ton\toff$ & $\toff\ton\toff\toff$ & $\toff\ton\toff\toff\toff$ & $\toff\ton\toff\toff\toff\ton$ \\
	      \hline
	      \cellcolor[gray]{0.9}$\repr{w^i}{\hypm}$ & $\varepsilon$ & $\ton$ & 
	      $\varepsilon$ & $\ton$ & $\varepsilon$ & $\ton$ & $\varepsilon$ \\
	      \hline
	      \cellcolor[gray]{0.9}$\suffx{w}{i}$ & $\toff\ton\toff\toff\toff\ton$ & 
	      $\rouge[\ton\toff\toff\toff\ton]$ & $\rouge[\toff\toff\toff\ton]$ & 
	      $\rouge[\toff\toff\ton]$ & $\toff\ton$ & $\ton$ & $\varepsilon$ \\
	      \hline
	      \cellcolor[gray]{0.9}$\pcomp[\obstab_m]{w}{i}$ & 
	      $\toff\ton\toff\toff\toff\ton$ & 
	      $\ton\ton\toff\toff\toff\ton$ & $\toff\toff\toff\ton$ & 
	      $\ton\toff\toff\ton$ & $\toff\ton$ & $\ton\ton$ & $\varepsilon$ \\
	      \hline
	      \cellcolor[gray]{0.9}$\agree{w}{i}$ & $\bleu[0]$ & $\bleu[\square]$ & 
	      $\bleu[\square]$ & $\bleu[1]$ & $\square$ & $1$ & $1$ \\
	      \hline
	    \end{tabular}}
  \end{center}
  
  Yet $\bleu[\interv{0}{3}]$ is a BI.
  We add suffixes 
  $\rouge[\ton\toff\toff\toff\ton]$, $\rouge[\toff\toff\toff\ton]$, and 
  $\rouge[\toff\toff\ton]$ to $S$.
\end{example}

\begin{theorem}
  \label{th:refine_simple_ce}
  Let $m$ be a model for the set of clauses $\clauses_{P, S}$, $w$ a 
  counterexample to $\hypm$, $\interv{i}{j}$ a breaking interval of 
  $w$, and $S' = S \cup \{\suffx{w}{i+1}, \ldots, \suffx{w}{j}\}$.
  Then any model $m'$ subsuming $m$ does not satisfy the updated set of
  clauses $\clauses_{P, S'}$.
\end{theorem}

Intuitively, assume that there exists an OT $\obstab_{m'}$ synthesized from a 
model $m'$ of $\clauses_{P, S'}$ that shares $\obstab_m$'s basis and membership 
assumptions—resulting in the same hypothesis being generated and $w$ still 
being a counterexample.
By definition of $S'$, partial computations
$\pcomp[\obstab_{m'}]{w}{k} = \repr{w^k}{\obstab_{m'}} \cdot \suffx{w}{k}$ 
belong to $\obstab_{m'}$'s cells for any $k \in \interv{i+1}{j}$.
Therefore, $m'$ must assign a value to predicate $\agree{w}{k}$.
But, due to $\agree{w}{i} \neq \agree{w}{j}$,
$\interv{i}{j}$ must contain a BP that causes a closure defect invalidating 
$m' \vDash \clauses_{P, S'}$.
By contradiction, no such $m'$ exists.

Property~\ref{prop:bp_bounds} guarantees a BI always exists.
Its proof and that of Theorem~\ref{th:refine_simple_ce} are detailed in
\appref{B.1}{prf:bp_bounds} and \appref{B.2}{prf:refine_simple_ce}.
Finding a BI requires at most $\calO(|w|)$ MQs, a weaker 
result than the $\bigo(\log|w|)$ MQs needed under the MAT framework:
binary search is no longer an efficient option due to $\square$ answers.
  
\paragraph{Use of simple counterexamples in $\linda$.}
Given a simple counterexample $w$ to hypothesis $\hypm$, we find a BI 
$\interv{i}{j}$ and add $\{\suffx{w}{i+1}, \ldots, \suffx{w}{j}\}$ to $S$,
resulting in new columns being inserted in the IOT $\obstab$.
The set of clauses $\clauses_{P, S}$ is then updated \defn{incrementally}:
new clauses are added to further constrain the current instance of the SAT 
solver without resetting it.
Depending on how MQs are carried out, we may be looking for the shortest 
BI or for one that maximizes $i$ (hence, minimizes the length 
of the longest suffix added).
We experimentally evaluate these two options in Section~\ref{sec:exp},
in particular in the context of RMC.

\subsection{Analysing inductive counterexamples}
\label{sec:ice_handling}

We briefly show how to handle inductive counterexamples to hypothesis $\hypm$
of the form $w = (w_1, w_2)$ in a similar fashion to simple counterexamples, 
inducing a refinement of $\hypm$ by \emph{frugally} expanding $S$.
Intuitively, we introduce a \defn{two-dimensional} incomplete evaluation 
predicate $\agreetd{w}{i}{j}$, querying the teacher and the 
model to determine whether $(\pcomp{w_1}{i}, \pcomp{w_2}{j})$ is an inductive 
pair for $(i, j) \in \interv{0}{|w_1|} \times \interv{0}{|w_2|}$.
We look for a \defn{breaking rectangle} (BR)
$\interv{i}{i'} \times \interv{j}{j'}$ s.t.
$\{\agreetd{w}{i}{j}, \agreetd{w}{i'}{j'}\} \subseteq \{0, 1\}$, and
$\agreetd{w}{i}{j} \neq \agreetd{w}{i'}{j'}$,
from which we deduce which suffixes of $w_1$ and $w_2$ to add to 
$S$.
  
\begin{theorem}
  \label{th:refine_inductive_ce}
  Let $m$ be a model for the set of clauses $\clauses_{P, S}$, $w = (w_1, 
  w_2)$ an inductive counterexample to $\hypm$,
  $\interv{i}{i'} \times \interv{j}{j'}$ a breaking rectangle,
  and $S' = S \cup {\{\suffx{w_1}{i}, \ldots, 
    \suffx{w_1}{i\nospace{'}}\} \cup \{\suffx{w_2}{j}, \ldots, 
    \suffx{w_2}{j\nospace{'}}\}}$.
  Then any model $m'$ subsuming $m$ does not satisfy the updated set
  of clauses $\clauses_{P, S'}$.
\end{theorem}

We refer the reader to \appref{C}{app:inductive_cex}
for formal definitions and proofs.
Finding a suitable BR requires at most
$\calO(|w_1| \times |w_2|)$ MQs.
  
\begin{example}
  Example~\ref{ex:equidist_hyp} featured an inductive hypothesis s.t. 
  $\post(H) \subseteq H$.
  \begin{center}
    {\begin{minipage}{0.56\linewidth}    
      However, its IOT can also yield another hypothesis with three states
      $\rouge[B = \{\varepsilon, \toff, \ton\}]$
      if we choose to ``fill in the blanks'' differently:
      
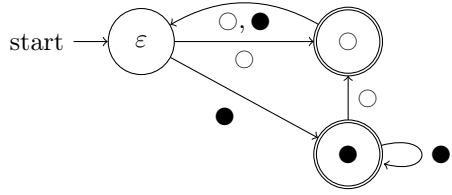
\begin{center}
        \begin{tikzpicture}[x=2.75cm, y=1.5cm]
          \node[state, initial] (q0) at (0,0) {$\varepsilon$};
          \node[state, accepting] (q1) at (1,0) {$\toff$};
          \node[state, accepting] (q2) at (1,-1) {$\ton$};
          
          \path (q0) edge[->] node[below] {$\toff$} (q1);
          \path (q0) edge[->] node[below left] {$\ton$} (q2);
          \path (q1) edge[->, bend right] node[below] {$\toff, \ton$} (q0);
          \path (q2) edge[->, loop right] node[right] {$\ton$} (q2);
          \path (q2) edge[->] node[right] {$\toff$} (q1);
        \end{tikzpicture}
      \end{center}
    \end{minipage}
    \hfill
    \begin{minipage}{0.4\linewidth}
      {\renewcommand{\arraystretch}{1.1}
      \setlength{\tabcolsep}{6pt}
      \begin{tabular}{|c|c|ccc|}
        \hline
        \multicolumn{2}{|c|}{} & 
        \multicolumn{3}{c|}{\cellcolor[gray]{0.9}$\mathbf{S}$} \\
        \cline{3-5}
        \multicolumn{2}{|c|}{} & $\varepsilon$ & $\ton$ & $\toff\ton$ \\
        \hline
        \cellcolor[gray]{0.9} & $\rouge[\varepsilon]$ & 
        $\fbox{\scriptsize{0}}$ & $\fbox{\scriptsize{1}}$ & 
        $\fbox{\scriptsize{0}}$ \\  
        \cellcolor[gray]{0.9} & $\rouge[\toff]$ & 
        $\fbox{\scriptsize{1}}$ & $\fbox{\scriptsize{0}}$ & 
        $\fbox{\scriptsize{1}}$ \\
        \multirow{-3}{*}{\cellcolor[gray]{0.9}$\mathbf{P}$} & 
        $\rouge[\ton]$ & $\fbox{\scriptsize{1}}$ & $1$ & $0$ \\
        \hline
        \cellcolor[gray]{0.9} & $\toff\toff$ & 
        $\fbox{\scriptsize{0}}$ & $\fbox{\scriptsize{1}}$ & 
        $\fbox{\scriptsize{0}}$ \\
        \cellcolor[gray]{0.9} & $\toff\ton$ & 
        $\fbox{\scriptsize{0}}$ & $\fbox{\scriptsize{1}}$ & $0$ \\
        \cellcolor[gray]{0.9} & $\ton\toff$ & 
        $\fbox{\scriptsize{1}}$ & $0$ & $1$ \\
        \multirow{-4}{*}{\cellcolor[gray]{0.9}$\mathbf{P \cdot \Sigma}$}
        & $\ton\ton$ & $\fbox{\scriptsize{1}}$ & 
        $\fbox{\scriptsize{1}}$ & $\fbox{\scriptsize{0}}$ \\
        \hline
      \end{tabular}}
    \end{minipage}}
  
    \bigskip
    {\begin{minipage}{0.5\linewidth}
      {\renewcommand{\arraystretch}{1.1}
      \setlength{\tabcolsep}{6pt}
      \begin{tabular}{|c|c|cccccc|}
        \hline
        \multicolumn{2}{|c|}{\cellcolor[gray]{0.9}} & 
        \multicolumn{6}{c|}{\cellcolor[gray]{0.9}$k$} \\
        \cline{3-8}
        \multicolumn{2}{|c|}{\multirow{-2}{*}{\cellcolor[gray]{0.9}
        $\agreetd{w}{k}{l}$}} & $0$ & $1$ & $2$ & $3$ & $4$ & $\rouge[5]$ \\
        \hline
        \cellcolor[gray]{0.9} & $0$ & $1$ & $1$ & $\square$ & $1$ & $1$ & $1$ 
        \\
        \cellcolor[gray]{0.9} & $\rouge[1]$ & $1$ & $1$ & $\square$ &$1$ & 
        $1$ & $\bleu[1]$ \\
        \cellcolor[gray]{0.9} & $\rouge[2]$ & $\square$ & $\square$ & 
        $\square$ & $\square$ & $\square$ & $\bleu[\square]$ \\
        \cellcolor[gray]{0.9} & $\rouge[3]$ & $\square$ & $\square$ & 
        $\square$ & $\square$ & $\square$ & $\bleu[\square]$ \\
        \cellcolor[gray]{0.9} & $\rouge[4]$ & $0$ & $0$ & $\square$ & $0$ & 
        $0$ & $\bleu[0]$ \\
        \multirow{-6}{*}{\cellcolor[gray]{0.9}$l$} & $5$ & $0$ & $0$ 
        & $\square$ & $0$ & $0$ & $0$ \\
        \hline
      \end{tabular}}
    \end{minipage}
    \hfill
    \begin{minipage}{0.46\linewidth}
      This hypothesis is not inductive: $(w_1, w_2) = 
      (\ton\toff\toff\ton\toff, \toff\ton\ton\toff\toff)$ is an inductive pair
      by Example~\ref{exple:equidist_indpair}, yet $w_1 \in H$ and
      $w_2 \not\in H$.
      It yields the 2D RS table to the left.
      \bleu[Rectangle] $\rouge[{[5] \times \interv{1}{4}}]$ is breaking;
      thus, we add suffixes $\{\suffx{w_1}{5}, \suffx{w_2}{1}, 
      \suffx{w_2}{2}, \suffx{w_2}{3}, \suffx{w_2}{4}\} = \{\varepsilon, 
      \ton\ton\toff\toff, \ton\toff\toff, \toff\toff, \toff\}$ to $S$.
    \end{minipage}}
  \end{center}
\end{example}

\paragraph{Use of inductive counterexamples in $\linda$.}
Given an inductive counterexample $(w_1, w_2)$ to hypothesis $\hypm$, we find 
a breaking rectangle $\interv{i}{i'} \times \interv{j}{j'}$
then update $S$ (hence the IOT $\obstab$ as well) with
${\{\suffx{w_1}{i}, \ldots, \suffx{w_1}{i\nospace{'}}\} \cup 
\{\suffx{w_2}{j}, \ldots, \suffx{w_2}{j\nospace{'}}\}}$.
Again, the resulting new clauses are added to the solver incrementally.
We may be looking for the smallest rectangle (in terms of area) or one that
maximizes $(i, j)$.

\subsection{Termination and correctness.}
  
Under the MAT framework for a unique target $L \in \Rat(\Sigma)$,
$\lstar$ maintains a set $P$ of words that are pairwise distinguished
under $\lequiv$.
This invariant guarantees termination and correctness:
to fix a closure defect, one adds a new element to $P$;
the algorithm ends when $P$ is a state cover of $L$.

However, $\linda$'s target class may contain more than one regular language.
Nevertheless, we can prove that for any $T \in \tgt$,
there exists a subset $B \subseteq P$ of words pairwise distinguished under 
${\equiv}_T$ that grows with core-induced refinements,
guaranteeing that targets of index $n+1$ can be learnt after at most $n$ 
refinements.
For any variable $v$ appearing in the algorithm,
let $\iter{v}{n}$ denote its value after the $n$-th UNSAT core analysis.
  
\begin{lemma}[\appref{D.1}{prf:progress}]
\label{lem:progress}
  Let $T \in \tgt$.
  Assume that $\linda$ has not ended after $n \in \integers$ UNSAT core 
  refinements.
  Then there exists a prefix-closed set $B_n \subseteq \iter{P}{n}$ 
  of size $n+1$ such that $\forall p, q \in B_n$,
  if $p \neq q$ then $p \not\equiv_T q$.
\end{lemma}

\begin{theorem}[\appref{D.2}{prf:termination}]
  \label{thm:termination}
  Assume $\tgt$ contains a language $T \in \tgt$ s.t. $\equiv_T$ has index
  $n \in \integers \setminus \{0\}$.
  Then $\linda$ terminates after at most $n-1$ UNSAT core refinements.
\end{theorem}

Core-based refinements may add more than one prefix:
hence, a target with $n+1$ states might require less than $n$ refinements to be
learnt.
In the worst case scenario, the algorithm may add up to $|P| \times |\Sigma|$
prefixes each refinement, leading to an exponential number of prefixes.
This is less likely to happen if some structure of $\tgt$ can be inferred from 
MQs—that is, if the teacher does not always return incomplete answers.

\section{Experiments}
\label{sec:exp}

We implemented~\footnote{\url{https://gitlab.lre.epita.fr/aa/string-chc-lib}}
{\linda} in Python using the PySAT 
\texttt{1.8.dev16}~\cite{imms-sat18,itk-sat24} library, as it allows us to tap
into the incremental SAT solver CaDiCaL \texttt{1.95}~\cite{Biere24}, a choice
motivated by the results of the
SAT competition~\cite{SATComp23}.
Our library also features basic automata manipulation, including
language separation and RMC-related operations for the \idmat{} implementation.

Each run is performed on an Ubuntu 24.04 LTS virtual machine with 8
vCPU and 16GB of RAM.
A run is considered failed when it timeouts after
$10$ minutes, or when it uses more than $200$~MB of RAM.
Since SAT solving and
UNSAT core extraction are non-deterministic processes, we resort to running
every experiment $3$ times then consider the best-of-three (Bo3) runtime and 
query complexity.

These experiments aim at addressing the following research questions:
\begin{description}
  \item[\Rq{1}] What is the impact of using a single SAT solver without 
  forking or maintaining a learner queue—that is, using {\linda} instead of 
  $L^*_\square$?
  \item[\Rq{2}] Does the introduction of inductive queries results in more 
  instances being solved within the allotted time-frame?
  \item[\Rq{3}] What is the most efficient way to introduce new suffixes during
  counterexample analysis?
\end{description}

\subsection{Language separation benchmark}
    
To address \Rq{1}, we consider the Nerode tool from~\cite{MWSKF23}
that implements the $L^*_\square$ algorithm mentioned in Section~\ref{sec:related}.
Nerode is compiled from its
repository~\footnote{\url{https://github.com/cornell-pl/nerode-public/commit/9d32ef89513297554db5aa6811d977acfc96e919}}
and ran with CLI options \texttt{-s -pq} on the same hardware,
with the same time and memory constraints.
The benchmark set does not consider inductive queries as they are
not supported by Nerode and focus instead here on the language separation 
problem of
Oliveira and al.~\cite{OliveiraS01} as formatted by~\cite{MWSKF23-artifact}.

Individual instances running times are reported comparatively in
Figure~\ref{fig:runsepvs}:
Nerode and \linda{} solved 96\% and 98\% of presented instances with
an average runtime of 24s and 6s, respectively.
We therefore conclude that the use of {\linda} is beneficial to the
overall learning time. We can also remark that the improvement could be even
greater as \linda{} is programmed in an interpreted language
while Nerode is written and compiled in OCaml.

\subsection{Inductive queries: an ablation study}
\label{sec:exp_indq}

To answer \Rq{2}, we now consider 28 RMC models from two different sources:
parameterized protocols from the RMC benchmark from \cite{ChenHLR17},
in a custom automata / transducer text file; and
handcrafted token passing protocols, in Python format, inspired
by Example~\ref{ex:equidist}.
    
As a baseline, we consider a \texttt{strict} (resp. \texttt{non-strict}) 
teacher—as defined in Section~\ref{sec:related}—that targets
$\post^*(S_0)$ (resp. $\pre^*(S_b)$). 
This setting did not allow us to achieve  
comparable~\footnote{We conjecture that this is due
to our Python implementation that exhaustively computes
$\post^*(S_0 \cap \Sigma^{n})$ and $\pre^*(S_b \cap \Sigma^n)$.}
performances to \cite{ChenHLR17};
thus, we perform an ablation study on \linda{} instead.
Table~\ref{tab:summary} displays the number of the instances solved and the 
Bo3 runtime / query complexity of the various AAL algorithms for the models 
they can solve within the allotted time-frame.
It shows that \linda{} instantiated with an inductive teacher is able to solve 
more models—with the RS optimization, all solvable models save one—than
\texttt{strict} and \texttt{non-strict} teachers thanks to the expressivity of 
the IdMAT framework, although it comes at the expense of runtime and query 
complexity compared to \texttt{non-strict}.

\subsection{Rivest-Schapire-like (RS) counterexample analysis \Rq{3}.}

We perform another ablation study on the dataset previously introduced in
Section~\ref{sec:exp_indq};
we compare the refinement method that consists in adding every suffix of a 
counterexample to $S$—in a similar fashion to $L^*_\square$~\cite{MWSKF23}—to 
our adaptation of the RS analysis technique under the {\idmat} framework
detailed in Sections~\ref{sec:ce_handling} and \ref{sec:ice_handling}.
Based on the type of BIs and BRs selected, two variants are considered:
do we seek to minimize the number of new suffixes added to $S$
(option \texttt{small}), or the length of the longest suffix added (option 
\texttt{short})?

Figure~\ref{fig:rsornotrs} compares the runtime of both RS variants with a 
RS-free execution of {\linda} and shows an overall improvement.
Table~\ref{tab:summary} also highlights a slight increase in the number of 
models solved when either variant of RS is enabled.
Variant \texttt{small} is shown to be faster than \texttt{short}:
in particular, it performs a significantly smaller amount of MQs.

\begin{figure}
  \centering
  \begin{subfigure}[ht]{0.46\textwidth}
    \begin{tikzpicture}
      \begin{axis}[
        ylabel={Nerode Avg. Runtime (s)},
        xlabel={\linda{} Avg. Runtime (s)},
               grid=major,
               legend pos=north west,
               ymode=log,
               xmode=log,
               ymin=0.1,
        xmax=600,
        ymax=600,
        width=\textwidth,
      ]
        \addplot[only marks, mark=x,blue] table[y expr={max(0.1,\thisrow{NerodeAvgTime})},
        x expr={max(0.1,\thisrow{LindaAvgTime})}, col sep=comma] {data/sep_vs.csv};
        \addplot[domain=0.001:600, dashed] {x};
      \end{axis}
    \end{tikzpicture}
    \caption{Language separation benchmark}
    \label{fig:runsepvs}
  \end{subfigure}
  \begin{subfigure}[ht]{0.46\textwidth}
    \begin{tikzpicture}
      \begin{axis}[
        xlabel={Runtime with RS (s)},
        ylabel={Runtime without RS (s)},
               grid=major,
               legend pos=south east,
               width=\textwidth,
               ymode=log,
               xmode=log,
      ]
        \addplot[only marks, mark=x, red] table[x=TimeRS, y=TimeNoRS, col sep=comma] {data/rmc.csv};
        \addlegendentry{\texttt{short}}
        \addplot[only marks, mark=o, blue] table[x=TimeNoshorten, y=TimeNoRS, col sep=comma] {data/rmc.csv};
        \addlegendentry{\texttt{small}}
        \addplot[domain=0.001:600, dashed] {x};
      \end{axis}
    \end{tikzpicture}
    \caption{Rivest-Shapire ablation study.}
    \label{fig:rsornotrs}
  \end{subfigure}
  \label{fig:runtimes}
  \caption{Runtime analysis and average values.}
\end{figure}

\begin{table}
  \begin{center}
    \renewcommand{\arraystretch}{1.25}
\scalebox{0.92}{\begin{tabular}{|c|c|c|c|c|c|c|c|}
\hline
& \linda{} \texttt{small}
& \linda{} w.o. RS
& \linda{} \texttt{short}
& \texttt{strict}
& \texttt{non-strict}
& Common
& V. Best
\\\hline
\#Models  
& 23
& 22
& 23
& 17
& 18
& 13
& 24
\\\hline
Solved freq.
& 0.79
& 0.75
& 0.80
& 0.61
& 0.64
& - 
& - 
\\\hline
Avg. Time
& 1.4
& 29.0
& 6.2
& 34.5
& 1.1
& - 
& - 
\\\hline
Avg. \#$\MEM$
& 203.4
& 733.6
& 644.2
& 591.4
& 146.1
& - 
& - 
\\\hline
Avg. \#$\VAL$
& 3.6
& 4.5
& 7.2
& 4.2
& 2.7
& - 
& - 
\\\hline
\end{tabular}}

  \end{center}
  \caption{Number of RMC models solved at least once.
  We average Bo3 query complexity and runtime on the set of common models every 
  algorithm can solve.}
  \label{tab:summary}
\end{table}

\section{Conclusion and Further Developments}

  We devised an inductive learning framework IdMAT that targets classes of 
  languages and subsumes the incomplete learning framework iMAT~\cite{MWSKF23}.
  Under IdMAT, we designed a new AAL algorithm $\linda$ that leverages 
  incremental SAT solving and UNSAT core analysis, as well as a new frugal 
  refinement method inspired by the well-known RS technique~\cite{RS93} that
  we also extended to inductive counterexamples.
  We implemented $\linda$ and benchmarked it in the context of regular language 
  separation and invariant synthesis.
 
  In the footsteps of~\cite{Leucker12,NeiderJ13,MWSKF23}, we based $\linda$ on 
  Angluin's $\lstar$~\cite{Angluin87}.
  However, it remains to be seen whether state-of-the-art active learning
  algorithms such as $L^\#$~\cite{Vaandrager22Lsharp} or
  $L^\lambda$~\cite{Howar22} can be extended to the IdMAT framework.
  The latter algorithm is of great interest to us as it tends 
  to reduce the overall \defn{symbol complexity}—the sum of the
  length of all words that have been queried—of the learning process:
  indeed, given a MQ on a word $w$ of length $n$,
  the oracle must exhaustively compute
  $\post^*({S_0} \cap \Sigma^{n})$ and $\pre^*({S_b} \cap \Sigma^{n})$ (see
  Section~\ref{sec:applications});
  the smaller the $n$, the faster the computation.
  Another promising lead to explore these sets consists in
  relying on GPUs and a matrix-based representation of the finite transducer 
  that encodes the system's transitions.

  Finally, we should also investigate the use of \defn{symbolic} alphabets, 
  as RTS are often partially or fully encoded symbolically~\cite{EsparzaRW22}.

\bibliographystyle{splncs04}
\bibliography{refs}

@article{IdMATTR,
  author       = {Daniel Stan and
                  Adrien Pommellet and
                  Juliette Jacquot},
  title        = {Automata Learning with an Incomplete but Inductive Teacher},
  journal      = {CoRR},
  volume       = {abs/2602.21073},
  year         = {2026},
  url          = {https://doi.org/10.48550/arXiv.2602.21073},
  doi          = {10.48550/ARXIV.2602.21073},
  eprinttype   = {arXiv},
  eprint       = {2602.21073},
  bibsource    = {dblp computer science bibliography, https://dblp.org}
}

@InProceedings{AKMG04,
author="Vardhan, Abhay
and Sen, Koushik
and Viswanathan, Mahesh
and Agha, Gul",
editor="Davies, Jim
and Schulte, Wolfram
and Barnett, Mike",
title="Learning to Verify Safety Properties",
booktitle="Formal Methods and Software Engineering",
year="2004",
publisher="Springer Berlin Heidelberg",
address="Berlin, Heidelberg",
pages="274--289",
isbn="978-3-540-30482-1"
}

@article{GS92,
  author =        {German, Steven M. and Sistla, A. Prasad},
  journal =       {Journal of the ACM},
  month =         jul,
  number =        {3},
  pages =         {675-735},
  publisher =     {ACM Press},
  title =         {Reasoning about Systems with Many Processes},
  volume =        {39},
  year =          {1992},
}

@inproceedings{Touili01,
  author       = {Tayssir Touili},
  editor       = {Richard Mayr},
  title        = {Regular Model Checking using Widening Techniques},
  booktitle    = {Verification of Parameterized Systems, {VEPAS} 2001, Satellite Workshop
                  of {ICALP} 2001, Crete, Greece, July 13, 2001},
  series       = {Electronic Notes in Theoretical Computer Science},
  volume       = {50},
  number       = {4},
  pages        = {342--356},
  publisher    = {Elsevier},
  year         = {2001},
  url          = {https://doi.org/10.1016/S1571-0661(04)00187-2},
  doi          = {10.1016/S1571-0661(04)00187-2},
  bibsource    = {dblp computer science bibliography, https://dblp.org}
}

@article{AbdullaHH16,
  author       = {Parosh Aziz Abdulla and
                  Fr{\'{e}}d{\'{e}}ric Haziza and
                  Luk{\'{a}}s Hol{\'{\i}}k},
  title        = {Parameterized verification through view abstraction},
  journal      = {Int. J. Softw. Tools Technol. Transf.},
  volume       = {18},
  number       = {5},
  pages        = {495--516},
  year         = {2016},
  url          = {https://doi.org/10.1007/s10009-015-0406-x},
  doi          = {10.1007/S10009-015-0406-X},
  bibsource    = {dblp computer science bibliography, https://dblp.org}
}

@article{Abdulla12,
  author       = {Parosh Aziz Abdulla},
  title        = {Regular model checking},
  journal      = {Int. J. Softw. Tools Technol. Transf.},
  volume       = {14},
  number       = {2},
  pages        = {109--118},
  year         = {2012},
  url          = {https://doi.org/10.1007/s10009-011-0216-8},
  doi          = {10.1007/S10009-011-0216-8},
  bibsource    = {dblp computer science bibliography, https://dblp.org}
}

@InProceedings{Jonsson00,
author="Jonsson, Bengt
and Nilsson, Marcus",
editor="Graf, Susanne
and Schwartzbach, Michael",
title="Transitive Closures of Regular Relations for Verifying Infinite-State Systems",
booktitle="Tools and Algorithms for the Construction and Analysis of Systems",
year="2000",
publisher="Springer Berlin Heidelberg",
address="Berlin, Heidelberg",
pages="220--235",
isbn="978-3-540-46419-8"
}

@article{LiffitonS08,
  author       = {Mark H. Liffiton and
                  Karem A. Sakallah},
  title        = {Algorithms for Computing Minimal Unsatisfiable Subsets of Constraints},
  journal      = {J. Autom. Reason.},
  volume       = {40},
  number       = {1},
  pages        = {1--33},
  year         = {2008},
  url          = {https://doi.org/10.1007/s10817-007-9084-z},
  doi          = {10.1007/S10817-007-9084-Z},
  bibsource    = {dblp computer science bibliography, https://dblp.org}
}

@inproceedings{imms-sat18,
  author    = {Alexey Ignatiev and Antonio Morgado and Joao Marques{-}Silva},
  title     = {{PySAT:} {A} {Python} Toolkit for Prototyping with {SAT} Oracles},
  booktitle = {SAT},
  pages     = {428--437},
  year      = {2018},
  url       = {https://doi.org/10.1007/978-3-319-94144-8_26},
  doi       = {10.1007/978-3-319-94144-8_26}
}

@inproceedings{itk-sat24,
  author    = {Alexey Ignatiev and Zi Li Tan and Christos Karamanos},
  title     = {Towards Universally Accessible {SAT} Technology},
  booktitle = {SAT},
  pages     = {4:1--4:11},
  year      = {2024},
  url       = {https://doi.org/10.4230/LIPIcs.SAT.2024.16},
  doi       = {10.4230/LIPICS.SAT.2024.16},
}

@phdthesis{Neider14,
  author    = {Daniel Neider},
  title     = {Applications of automata learning in verification and synthesis},
  school    = {{RWTH} Aachen University},
  year      = {2014},
  url       = {http://darwin.bth.rwth-aachen.de/opus3/volltexte/2014/5169},
  urn       = {urn:nbn:de:hbz:82-opus-51699},
  bibsource = {dblp computer science bibliography, https://dblp.org}
}

@inproceedings{NeiderJ13,
  author       = {Daniel Neider and
                  Nils Jansen},
  editor       = {Guillaume Brat and
                  Neha Rungta and
                  Arnaud Venet},
  title        = {Regular Model Checking Using Solver Technologies and Automata Learning},
  booktitle    = {{NASA} Formal Methods, 5th International Symposium, {NFM} 2013, Moffett
                  Field, CA, USA, May 14-16, 2013. Proceedings},
  series       = {Lecture Notes in Computer Science},
  volume       = {7871},
  pages        = {16--31},
  publisher    = {Springer},
  year         = {2013},
  url          = {https://doi.org/10.1007/978-3-642-38088-4\_2},
  doi          = {10.1007/978-3-642-38088-4\_2},
  bibsource    = {dblp computer science bibliography, https://dblp.org}
}

@inproceedings{ChenHLR17,
  author    = {Yu{-}Fang Chen and
               Chih{-}Duo Hong and
               Anthony W. Lin and
               Philipp R{\"{u}}mmer},
  editor    = {Daryl Stewart and
               Georg Weissenbacher},
  title     = {Learning to prove safety over parameterised concurrent systems},
  booktitle = {2017 Formal Methods in Computer Aided Design, {FMCAD} 2017, Vienna,
               Austria, October 2-6, 2017},
  pages     = {76--83},
  publisher = {{IEEE}},
  year      = {2017},
  url       = {https://doi.org/10.23919/FMCAD.2017.8102244},
  doi       = {10.23919/FMCAD.2017.8102244},
  bibsource = {dblp computer science bibliography, https://dblp.org}
}

@InProceedings{czerner2024,
  author =	{Czerner, Philipp and Esparza, Javier and Krasotin, Valentin and Welzel-Mohr, Christoph},
  title =	{{Computing Inductive Invariants of Regular Abstraction Frameworks}},
  booktitle =	{35th International Conference on Concurrency Theory (CONCUR 2024)},
  pages =	{19:1--19:18},
  series =	{Leibniz International Proceedings in Informatics (LIPIcs)},
  ISBN =	{978-3-95977-339-3},
  ISSN =	{1868-8969},
  year =	{2024},
  volume =	{311},
  editor =	{Majumdar, Rupak and Silva, Alexandra},
  publisher =	{Schloss Dagstuhl -- Leibniz-Zentrum f{\"u}r Informatik},
  address =	{Dagstuhl, Germany},
  URL =		{https://drops.dagstuhl.de/entities/document/10.4230/LIPIcs.CONCUR.2024.19},
  URN =		{urn:nbn:de:0030-drops-207919},
  doi =		{10.4230/LIPIcs.CONCUR.2024.19}
}

@inproceedings{MWSKF23,
  author       = {Mark Moeller and
                  Thomas Wiener and
                  Alaia Solko{-}Breslin and
                  Caleb Koch and
                  Nate Foster and
                  Alexandra Silva},
  editor       = {Karim Ali and
                  Guido Salvaneschi},
  title        = {Automata Learning with an Incomplete Teacher},
  booktitle    = {37th European Conference on Object-Oriented Programming, {ECOOP} 2023,
                  July 17-21, 2023, Seattle, Washington, United States},
  series       = {LIPIcs},
  volume       = {263},
  pages        = {21:1--21:30},
  publisher    = {Schloss Dagstuhl - Leibniz-Zentrum f{\"{u}}r Informatik},
  year         = {2023},
  url          = {https://doi.org/10.4230/LIPIcs.ECOOP.2023.21},
  doi          = {10.4230/LIPICS.ECOOP.2023.21},
  bibsource    = {dblp computer science bibliography, https://dblp.org}
}

@Article{MWSKF23-artifact,
  author =	{Moeller, Mark and Wiener, Thomas and Solko-Breslin, Alaia and Koch, Caleb and Foster, Nate and Silva, Alexandra},
  title =	{{Automata Learning with an Incomplete Teacher (Artifact)}},
  pages =	{21:1--21:3},
  journal =	{Dagstuhl Artifacts Series},
  ISSN =	{2509-8195},
  year =	{2023},
  volume =	{9},
  number =	{2},
  editor =	{Moeller, Mark and Wiener, Thomas and Solko-Breslin, Alaia and Koch, Caleb and Foster, Nate and Silva, Alexandra},
  publisher =	{Schloss Dagstuhl -- Leibniz-Zentrum f{\"u}r Informatik},
  address =	{Dagstuhl, Germany},
  URL =		{https://drops.dagstuhl.de/entities/document/10.4230/DARTS.9.2.21},
  URN =		{urn:nbn:de:0030-drops-182612},
  doi =		{10.4230/DARTS.9.2.21}
}

@article{RS93,
  author       = {Ronald L. Rivest and
  Robert E. Schapire},
  title        = {Inference of Finite Automata Using Homing Sequences},
  journal      = {Inf. Comput.},
  volume       = {103},
  number       = {2},
  pages        = {299--347},
  year         = {1993},
  url          = {https://doi.org/10.1006/inco.1993.1021},
  doi          = {10.1006/INCO.1993.1021},
  bibsource    = {dblp computer science bibliography, https://dblp.org}
}

@article{Angluin87,
  author       = {Dana Angluin},
  title        = {Learning Regular Sets from Queries and Counterexamples},
  journal      = {Inf. Comput.},
  volume       = {75},
  number       = {2},
  pages        = {87--106},
  year         = {1987},
  url          = {https://doi.org/10.1016/0890-5401(87)90052-6},
  doi          = {10.1016/0890-5401(87)90052-6},
  bibsource    = {dblp computer science bibliography, https://dblp.org}
}

@article{EsparzaRW25,
  author       = {Javier Esparza and
                  Michael Raskin and
                  Christoph Welzel{-}Mohr},
  title        = {Regular Model Checking Upside-Down: An Invariant-Based Approach},
  journal      = {Log. Methods Comput. Sci.},
  volume       = {21},
  number       = {1},
  pages        = {4},
  year         = {2025},
  url          = {https://doi.org/10.46298/lmcs-21(1:4)2025},
  doi          = {10.46298/LMCS-21(1:4)2025},
  bibsource    = {dblp computer science bibliography, https://dblp.org}
}

@inproceedings{EsparzaRW22,
  author       = {Javier Esparza and
                  Mikhail A. Raskin and
                  Christoph Welzel},
  editor       = {Bartek Klin and
                  Slawomir Lasota and
                  Anca Muscholl},
  title        = {Regular Model Checking Upside-Down: An Invariant-Based Approach},
  booktitle    = {33rd International Conference on Concurrency Theory, {CONCUR} 2022,
                  September 12-16, 2022, Warsaw, Poland},
  series       = {LIPIcs},
  volume       = {243},
  pages        = {23:1--23:19},
  publisher    = {Schloss Dagstuhl - Leibniz-Zentrum f{\"{u}}r Informatik},
  year         = {2022},
  url          = {https://doi.org/10.4230/LIPIcs.CONCUR.2022.23},
  doi          = {10.4230/LIPICS.CONCUR.2022.23},
  bibsource    = {dblp computer science bibliography, https://dblp.org}
}

@InProceedings{Leucker12,
  author = "Leucker, Martin
  and Neider, Daniel",
  editor = "Margaria, Tiziana
  and Steffen, Bernhard",
  title = "Learning Minimal Deterministic Automata from Inexperienced Teachers",
  booktitle = "Leveraging Applications of Formal Methods, Verification and Validation. Technologies for Mastering Change",
  year = "2012",
  publisher = "Springer Berlin Heidelberg",
  address = "Berlin, Heidelberg",
  pages = "524--538",
  isbn = "978-3-642-34026-0"
}

@InProceedings{Grinchtein06,
  author = "Grinchtein, Olga
  and Leucker, Martin
  and Piterman, Nir",
  editor = "Furbach, Ulrich
  and Shankar, Natarajan",
  title = "Inferring Network Invariants Automatically",
  booktitle = "Automated Reasoning",
  year = "2006",
  publisher = "Springer Berlin Heidelberg",
  address = "Berlin, Heidelberg",
  pages = "483--497",
  isbn = "978-3-540-37188-5"
}

@InProceedings{Chen09,
  author = "Chen, Yu-Fang
  and Farzan, Azadeh
  and Clarke, Edmund M.
  and Tsay, Yih-Kuen
  and Wang, Bow-Yaw",
  editor = "Kowalewski, Stefan
  and Philippou, Anna",
  title = "Learning Minimal Separating DFA's for Compositional Verification",
  booktitle = "Tools and Algorithms for the Construction and Analysis of Systems",
  year = "2009",
  publisher = "Springer Berlin Heidelberg",
  address = "Berlin, Heidelberg",
  pages = "31--45",
  isbn = "978-3-642-00768-2"
}

@book{SATComp23,
  title = "Proceedings of SAT Competition 2023: Solver, Benchmark and Proof Checker Descriptions",
  editor = "Tomas Balyo and Marijn Heule and Markus Iser and Matti J{\"a}rvisalo and Martin Suda",
  year = "2023",
  language = "English",
  series = "Department of Computer Science Series of Publications B",
  publisher = "Department of Computer Science, University of Helsinki",
  address = "Finland",
}

@InProceedings{Biere24,
  author = "Biere, Armin
  and Faller, Tobias
  and Fazekas, Katalin
  and Fleury, Mathias
  and Froleyks, Nils
  and Pollitt, Florian",
  editor = "Gurfinkel, Arie
  and Ganesh, Vijay",
  title = "CaDiCaL 2.0",
  booktitle = "Computer Aided Verification",
  year = "2024",
  publisher = "Springer Nature Switzerland",
  address = "Cham",
  pages = "133--152",
  isbn = "978-3-031-65627-9"
}

@inproceedings{Howar22,
  author       = {Falk Howar and
  Bernhard Steffen},
  editor       = {Nils Jansen and
  Mari{\"{e}}lle Stoelinga and
  Petra van den Bos},
  title        = {Active Automata Learning as Black-Box Search and Lazy Partition Refinement},
  booktitle    = {A Journey from Process Algebra via Timed Automata to Model Learning
  - Essays Dedicated to Frits Vaandrager on the Occasion of His 60th
  Birthday},
  series       = {Lecture Notes in Computer Science},
  volume       = {13560},
  pages        = {321--338},
  publisher    = {Springer},
  year         = {2022},
  url          = {https://doi.org/10.1007/978-3-031-15629-8\_17},
  doi          = {10.1007/978-3-031-15629-8\_17},
  bibsource    = {dblp computer science bibliography, https://dblp.org}
}

@InProceedings{Vaandrager22Lsharp,
  author = "Vaandrager, Frits
  and Garhewal, Bharat
  and Rot, Jurriaan
  and Wi{\ss}mann, Thorsten",
  editor = "Fisman, Dana
  and Rosu, Grigore",
  title = "A New Approach for Active Automata Learning Based on Apartness",
  booktitle = "Tools and Algorithms for the Construction and Analysis of Systems",
  year = "2022",
  publisher = "Springer International Publishing",
  address = "Cham",
  pages = "223--243",
  isbn = "978-3-030-99524-9"
}

@article{OliveiraS01,
  author       = {Arlindo L. Oliveira and
                  Jo{\~{a}}o P. Marques Silva},
  title        = {Efficient Algorithms for the Inference of Minimum Size DFAs},
  journal      = {Mach. Learn.},
  volume       = {44},
  number       = {1/2},
  pages        = {93--119},
  year         = {2001},
  url          = {https://doi.org/10.1023/A:1010828029885},
  doi          = {10.1023/A:1010828029885},
  bibsource    = {dblp computer science bibliography, https://dblp.org}
}

@Inbook{Steffen2011,
  author="Steffen, Bernhard
  and Howar, Falk
  and Merten, Maik",
  editor="Bernardo, Marco
  and Issarny, Val{\'e}rie",
  title="Introduction to Active Automata Learning from a Practical Perspective",
  bookTitle="Formal Methods for Eternal Networked Software Systems: 11th International School on Formal Methods for the Design of Computer, Communication and Software Systems, SFM 2011, Bertinoro, Italy, June 13-18, 2011. Advanced Lectures",
  year="2011",
  publisher="Springer Berlin Heidelberg",
  address="Berlin, Heidelberg",
  pages="256--296",
  isbn="978-3-642-21455-4",
  doi="10.1007/978-3-642-21455-4_8",
  url="https://doi.org/10.1007/978-3-642-21455-4_8"
}

@inproceedings{DBLP:conf/models/YaacovWAH25,
  author       = {Tom Yaacov and
  Gera Weiss and
  Gal Amram and
  Avi Hayoun},
  title        = {Automata Models for Effective Bug Pattern Description},
  booktitle    = {28th {ACM/IEEE} International Conference on Model Driven Engineering
  Languages and Systems, {MODELS} 2025, Grand Rapids, MI, USA, October
  5-10, 2025},
  pages        = {119--129},
  publisher    = {{IEEE}},
  year         = {2025},
  url          = {https://doi.org/10.1109/MODELS67397.2025.00017},
  doi          = {10.1109/MODELS67397.2025.00017},
  bibsource    = {dblp computer science bibliography, https://dblp.org}
}

\ifcameraready
\else
\newpage
\appendix

\section{SAT Encoding of the Hypothesis}
\label{app:sat_clauses}

For convenience's sake, we introduce the \defn{value} $\val$ operator such 
that, if $m$ is a model,
\begin{enumerate*}
  \item $\val_m(w) = \MEM(w, \emptyset)$
  if $\MEM(w, \emptyset) \in \{0, 1\}$, else
  \item $\val_m(w) = m[x_w]$ if $m[x_w]$ is defined, else
  \item $\val_m(w) = x_w$ and the answer is then said to be \defn{incomplete}.
\end{enumerate*}
We denote $\val = \val_\emptyset$.
Intuitively, $\val_m(w)$ first queries the teacher, then model $m$, then 
introduces a membership variable $x_w$.
Figure \ref{fig:clauses} lists the clauses $\clauses_{P, S}$ of our encoding.
We assume a total ordering $<$ on $\Sigma^*$.
\begin{figure}[h]
  {\renewcommand{\arraystretch}{1.25}
  \setlength{\tabcolsep}{7pt}
  \begin{tabular}{llc}
    \hline
    \textbf{Conditions on parameters} &
    \textbf{Formula} &
    \textbf{Name} \\
    \hline
    & $b_{\varepsilon}$ & $\cbasis{\varepsilon}$ \\
    $pa \in P$ & $b_{pa} \rightarrow b_p$ & $\cbasis{pa}$ \\
    $pa \in P$ & $b_{pa} \rightarrow e_{p, a, pa}$ & $\creach{pa}$ \\
    $a \in \Sigma, p,p' \in P, p'\neq pa, s\in S$ &
    $e_{p, a, p'} \rightarrow (\val(pas) \leftrightarrow \val(p's))$
    & $\ccong{p, a, p', s}$ \\
    $p, p_1, p_2 \in P, a\in \Sigma, p_1 < p_2$ &
    $e_{p, a, p_1} \rightarrow \neg e_{p, a, p_2}$ &
    $\cdet{p, a, p_1, p_2}$ \\
    $p\in P, a \in \Sigma$ & $b_p \rightarrow
    \bigvee_{p'\in P} e_{p, a, p'}$
    & $\cclos{p, a}$ \\
    $p, p' \in P, a \in \Sigma$ & $e_{p, a, p'} \rightarrow b_{p'}$ &
    $\csucc{p, a, p'}$ \\
    $pa, p' \in P, p' < pa$ & $b_{pa} \land b_{p'} \rightarrow
    \neg e_{p, a, p'}$ &
    $\csharp{p, a, p'}$ \\
    $(w_1, w_2) \in \indpairs$ & $x_{w_1} \rightarrow x_{w_2}$
    & $\cind{w_1, w_2}$ \\
    \hline
  \end{tabular}}
  \caption{The encoding $\clauses_{P, S}$ of the observation table.}
  \label{fig:clauses}
\end{figure}

These PFs can be trivially converted to CNF without exponential blow-up, 
resulting in $\bigo(|P|^3 |\Sigma| + |P|^2 |\Sigma|^2 |S|^2)$
clauses and $\bigo(|P|^2 |\Sigma| + |P| |S|)$ variables.
They are to be understood as follows:
\begin{description}
  \item[$\cbasis{\varepsilon}$] Word $\varepsilon$ always belongs to the   
  basis.
  \item[$\cbasis{pa}$] The basis must be prefix-closed.
  \item[$\creach{pa}$] If $p \cdot a$ belongs to the basis, it is the successor 
  of state $p$ by letter $a$.
  \item[$\ccong{p, a, p', s}$] If prefix $p$'s successor by letter $a$ is  
  $p'$, then $p \cdot a \equiv_{\obstab_m} p'$.
  \item[$\cdet{p, a, p_1, p_2}$] In order to generate a deterministic  
  hypothesis, state $p$ must admit an unique successor $p_1$ 
  induced by the equivalence relation $\equiv_{\obstab_m}$.
  \item[$\cclos{p, a}$] OT $\obstab_m$ must be closed:
  given a state $p$ in the basis, its successor by letter $a$ must be 
  $\equiv_{\obstab_m}$-equivalent to some $p' \in B$.
  \item[$\csucc{p, a, p'}$] Successors of a basis state must belong to the
  basis.
  \item[$\csharp{p, a, p'}$] OT $\obstab_m$ must be sharp:
  given $p', pa \in P$ s.t. $p' < pa$, $pa \not\equiv_{\obstab_m} p'$.
  \item[$\cind{w_1, w_2}$] The inductive constraints on $\cells$ must be 
  enforced.
\end{description}

\paragraph{Building the hypothesis.}
\label{sec:build_hyp}

Assume that $\clauses_{P, S}$ is satisfiable by a model $m$.
Then $m$ induces the following DFA $\hypm$:
\begin{itemize}
  \item Its state space is the basis $B_m = \{p \in P \mid m[b_p] = 1\}$.
  It is a prefix-closed subset of $P$ that contains $\varepsilon$
  thanks to clauses $\cbasis{\varepsilon}$ and $\cbasis{pa}$.
  \item Its initial state is $\varepsilon \in B_m$ and its set of final states 
  is $\{p \in B_m \mid \val_m(p) = 1\}$.
  \item For any $p \in B_m$ and $a \in \Sigma$, $\delta_m(p, a) = p'$,
  where $p'$ is the unique prefix in $B_m$ such that $m[e_{p, a, p'}] = 1$;
  it exists due to clauses $\cdet{p, a, p_1, p_2}$, $\cclos{p, a}$,
  and $\csucc{p, a, p'}$.
\end{itemize}
By design, the resulting hypothesis is \defn{$B_m$-compatible}:
$\forall p \in B_m$, $\hypm(p) = \val_m(p)$.

Note that if a prefix $p \in P$ does not belong to $B_m$,
there is no need to ``fill in the blanks" of lines $p$ and $p \cdot a$ of 
$\obstab$, $a \in \Sigma$,
due to said lines being likely to be dropped from $\obstab_m$.
Indeed, once $B_m$ is known, the only variables used to synthesize $\hypm$ are 
related to sub-table $\obstab_m$.
Thus, we may assume w.l.o.g. that for $a \in \Sigma$ and $s \in S$,
$m[e_{p, a, p'}]$, $m[x_{p \cdot s}]$ and $m[x_{pa \cdot s}]$ are only defined 
if $p, p' \in B_m$.

\paragraph{Analyzing the UNSAT core.}
\label{sec:unsat_core}

Should the instance $\clauses_{P, S}$ end up being UNSAT,
then its core always contains at least one clause $\cclos{p, a}$:
were all variables $e_{p, a, p'}$, $x_w$ and $b_p$ (save 
$b_\varepsilon$ that is set to $1$) set to $0$, then all the other clauses 
would be satisfied.
As discussed in Section~\ref{sec:fill_subtable}, we can use this core to 
refine $P$.

\section{Analyzing Simple Counterexamples}

\subsection{Proof of Property~\ref{th:refine_simple_ce}}
\label{prf:bp_bounds}

\begin{proof}
  Since $w$ is a simple counterexample, $\MEM(w) \in \{0, 1\}$ and
  $\hypm(w) \neq \MEM(w)$ \textbf{(\romannumeral 1)}.
  Moreover, $\agree{w}{0} = \MEM(w)$ \textbf{(\romannumeral 2)}
  due to $\pcomp{w}{0} = w$.
  
  Since $\pcomp{w}{|w|} = \repr{w}{\hypm}$, $\agree{w}{|w|} = 
  \val_m(\repr{w}{\hypm})$.
  But $\hypm(w) = \hypm(\repr{w}{\hypm})$ by definition of representatives, and
  $\hypm(\repr{w}{\hypm}) = \val_m(\repr{w}{\hypm})$ due to
  hypothesis $\hypm$ being $B_m$-compatible.
  Thus, $\agree{w}{|w|} = \hypm(w)$ \textbf{(\romannumeral 3)}.
  
  Finally, by \textbf{\romannumeral 1}, \textbf{\romannumeral 2}, and 
  \textbf{\romannumeral 3}, $\agree{w}{0} \neq \agree{w}{|w|}$.\hfill\qed
\end{proof}

\subsection{Proof of Theorem~\ref{th:refine_simple_ce}}
\label{prf:refine_simple_ce}

\begin{proof}
  Assume that there exists a model $m'$ subsuming $m$ and satisfying 
  $\clauses_{P, S'}$.
  Naturally, $m'$ induces the same basis $B$ and the same hypothesis $\hypm$. 
  Due to $\{\suffx{w}{i+1}, \ldots, \suffx{w}{j}\} \subseteq S'$, note that
  $c_{w, k}$ belongs to the updated set $\cells$ for any $k \in 
  \interv{i+1}{j}$, being the concatenation of $\repr{w}{k} \in P$ and 
  $\suffx{w}{k} \in S'$;
  thus, if $c_{w, k} \in U$ (see Definition~\ref{def:indpairs}), then 
  \textbf{(\romannumeral 1)} $m'[c_{w, k}]$ is defined.
  
  Word $w$ remains a counterexample to $\hypm$ and
  $\agree{w}{i} \neq \agree{w}{j}$ still holds.
  Moreover, $\agree{w}{k} \in \{0, 1\}$ for every $k \in \interv{i}{j}$ by
  \textbf{\romannumeral 1}.
  Thus, interval $\interv{i}{j}$ contains a breaking point:
  there always exists a $k \in \interv{i}{j-1}$ such that
  $\agree{w}{k} \neq \agree{w}{k+1}$, i.e.
  $\val_{m'}(\repr{w}{k} \cdot w[k] \cdot \suffx{w}{k+1}) \neq 
  \val_{m'}(\repr{w}{k+1} \cdot \suffx{w}{k+1})$.
  
  By clause $\ccong{\repr{w}{k}, w[k], \repr{w}{k+1}, \suffx{w}{k+1}}$,
  $m'[e_{\repr{w}{k}, w[k], \repr{w}{k+1}}] = 0$.
  But by design of $\hypm$, $\delta_{m'}(\repr{w}{k}, w[k]) = \repr{w}{k+1}$, 
  thus $m'[e_{\repr{w}{k}, w[k], \repr{w}{k+1}}] = 1$.
  Finally, there is a contradiction and no such model $m'$ exists.
  Intuitively, suffix $\suffx{w}{k+1}$ distinguishes
  $\repr{w}{k} \cdot w[k]$ and $\repr{w}{k+1}$ despite these two words being 
  equivalent w.r.t.  $\equiv_{\obstab_{m'}}$.\hfill\qed
\end{proof}

\section{Analyzing Inductive Counterexamples}
\label{app:inductive_cex}

\subsection{Definitions for a two-dimensional analysis}
\label{app:inductive_cex_defs}

\begin{definition}[Two-dimensional incomplete evaluation predicate]
  \label{def:agree2d}
  Let $(w_1, w_2)$ be an inductive counterexample to $\hypm$.
  For $(i, j) \in \interv{0}{|w_1|} \times \interv{0}{|w_2|}$,
  we define $\agreetd{w}{i}{j}$ as follows:
  \begin{enumerate}
    \item if $\MEM(\pcomp{w_1}{i}), \MEM(\pcomp{w_2}{j}) \in \{0, 1\}$, then
    $\agreetd{w}{i}{j} =
    \MEM(\pcomp{w_1}{i}) \rightarrow \MEM(\pcomp{w_2}{j})$,
    \item else if $(\pcomp{w_1}{i}, \pcomp{w_2}{j}) \in \MEM(\pcomp{w_1}{i},
    \{\pcomp{w_2}{j}\})$, then $\agreetd{w}{i}{j} = 1$,
    \item else if $m[x_{\pcomp{w_1}{i}}]$, $m[x_{\pcomp{w_2}{j}}]$ are defined, 
    then $\agreetd{w}{i}{j} =
    m[x_{\pcomp{w_1}{i}}] \rightarrow m[x_{\pcomp{w_2}{j}}]$,
    \item else $\agreetd{w}{i}{j} = \square$.
  \end{enumerate}
\end{definition}

\begin{property}
  \label{prop:inductive_bp_bounds}
  $\agreetd{w}{0}{0} = 1$ and $\agreetd{w}{|w_1|}{|w_2|} = 0$.
\end{property}

Let $\preceq$ be the partial order relation on
$\interv{0}{|w_1|} \times \interv{0}{|w_2|}$ such that
$(i,j) \preceq (i',j') \leftrightarrow (i \leq i') \land (j \leq j')$.
We consider \emph{rectangles} instead of intervals:

\begin{definition}
  A \defn{Breaking Rectangle} (BR) is a rectangle
  $\interv{i}{i'} \times \interv{j}{j'}$ such that
  $(i, j), (i', j') \in [0:|w_1|] \times [0:|w_2|]$,
  $(i, j) \prec (i', j')$,
  $\{\agreetd{w}{i}{j}, \agreetd{w}{i'}{j'}\} \subseteq \{0, 1\}$, and
  $\agreetd{w}{i}{j} \neq \agreetd{w}{i'}{j'}$.
\end{definition}

\subsection{Proof of Property~\ref{prop:inductive_bp_bounds}}
\label{prf:inductive_bp_bounds}

Property~\ref{prop:inductive_bp_bounds} guarantees that a BR always 
exists.

\begin{proof}
  Note that $\agreetd{w}{0}{0} = 1$ due to
  $(\pcomp{w_1}{0}, \pcomp{w_2}{0}) = (w_1, w_2)$
  being an inductive pair by definition.
  Moreover, $\hypm(w_i) = \val_m(\pcomp{w_i}{|w_i|})$ for $i \in \{1, 2\}$.
  But $\hypm(w_1) \not\rightarrow \hypm(w_2)$ as $(w_1, w_2)$ is an
  inductive counterexample.
  Thus
  $\val_m(\pcomp{w_1}{|w_1|}) \not\rightarrow \val_m(\pcomp{w_2}{|w_2|})$.
  and $\agreetd{w}{|w_1|}{|w_2|} = 0$.\hfill\qed
\end{proof}

\subsection{Proof of Theorem~\ref{th:refine_inductive_ce}}
\label{prf:refine_inductive_ce}

\begin{proof}
  Assume that there exists a model $m'$ subsuming $m$ and satisfying 
  $\clauses_{P, S'}$.
  Naturally, $m'$ induces the same basis $B$, the same hypothesis $\hypm$, and
  the same counterexample $w = (w_1, w_2)$. 
  Due to $\{\suffx{w_1}{i}, \ldots, \suffx{w_1}{i'}\} \cup
  \{\suffx{w_2}{j}, \ldots, \suffx{w_2}{j'}\} \subseteq S'$, note that
  \textbf{(\romannumeral 1)} $c_{w_1, k}$ (resp. $c_{w_2, l}$) belongs to
  the updated set $\cells$ for any $k \in \interv{i}{i'}$
  (resp. $l \in \interv{j}{j'}$);
  thus, if $c_{w_1, k} \in U$ (resp. $c_{w_2, l} \in U$), then
  \textbf{(\romannumeral 2)} $m'[c_{w_1, k}]$ (resp. $m'[c_{w_2, l}]$) is 
  defined.
  
  Note that $\agreetd{w}{i}{j} \neq \agreetd{w}{i'}{j'}$ still holds.
  Moreover, $\agreetd{w}{k}{l} \in \{0, 1\}$
  for every $(k,l) \in \interv{i}{i'} \times \interv{j}{j'}$ by 
  \textbf{\romannumeral 2}.
  Thus, rectangle $\interv{i}{i'} \times \interv{j}{j'}$ always contains 
  a breaking point.
  
  \begin{description}
    \item[Horizontal (1, 0).] Assume that is horizontal,
    i.e. $\exists k \in \interv{i}{i'-1}$ and $\exists l \in \interv{j}{j'}$ 
    s.t. $\agreetd{w}{k}{l} \neq \agreetd{w}{k+1}{l}$,
    and that $(\agreetd{w}{k}{l}, \agreetd{w}{k+1}{l}) = (1, 0)$.
    
    As $\agreetd{w}{k+1}{l} = 0$, case \textbf{1.} or \textbf{3.} of 
    Definition~\ref{def:agree2d} apply:
    \textbf{(\romannumeral 3)} $\val_{m'}(\pcomp{w_1}{k+1}) = 1$ and
    \textbf{(\romannumeral 4)} $\val_{m'}(\pcomp{w_2}{l}) = 0$, as this is
    the only way to negate the implication.
    
    Let us now study the consequences of $\agreetd{w}{k}{l} = 1$.
    \begin{itemize}
      \item If case \textbf{2.} of Definition~\ref{def:agree2d} applies, 
      $(\pcomp{w_1}{k}, \pcomp{w_2}{l})$ is an inductive pair.
      Since $\pcomp{w_1}{k}$ and $\pcomp{w_2}{l}$ both belong to $\cells$
      by \textbf{\romannumeral 1},
      $(\pcomp{w_1}{k}, \pcomp{w_2}{l}) \in \indpairs$;
      
      Suppose that $\val_{m'}(\pcomp{w_1}{k}) = 1$;
      then $\val_{m'}(\pcomp{w_2}{l}) = 1$
      due to $\cind{\pcomp{w_1}{k}, \pcomp{w_2}{l}}$.
      This contradicts \textbf{\romannumeral 4}.
      As a consequence, $\val_{m'}(\pcomp{w_1}{k}) = 0$.
      
      \item Assume case \textbf{1.} or case \textbf{3.} of 
      Definition~\ref{def:agree2d} apply;
      suppose again that $\val_{m'}(\pcomp{w_1}{k}) = 1$;
      the implication $\val_{m'}(\pcomp{w_1}{k}) \rightarrow 
      \val_{m'}(\pcomp{w_2}{l})$
      results in the value $\val_{m'}(\pcomp{w_2}{l}) = 1$,
      contradicting \textbf{\romannumeral 4}.
      As a consequence, $\val_{m'}(\pcomp{w_1}{k}) = 0$.
    \end{itemize}
    
    In both cases, we have proven that $\val_{m'}(\pcomp{w_1}{k}) = 0$.
    Finally, by \textbf{\romannumeral 3}, we have proven
    $\val_{m'}(\pcomp{w_1}{k}) \notleftrightarrow 
    \val_{m'}(\pcomp{w_1}{k+1})$.
    From then on, a reasoning similar to the proof of 
    Theorem~\ref{th:refine_simple_ce} demonstrates that $m'$ can't
    satisfy $\clauses_{P, S'}$.
    
    \item[Horizontal (0, 1), Vertical (1, 0), Vertical (0, 1).]
    These proofs are for all intents and purposes similar to 
    the previous case and we omit them.
    \end{description}
In every case, $m'$ violates the set of clauses $\clauses_{P, S'}$.\hfill\qed
\end{proof}

\section{Termination and Correctness of $\linda$}

\subsection{Proof of Lemma~\ref{lem:progress}}
\label{prf:progress}

\begin{proof}
  We use a proof by induction on $n$, the base case $n = 0$ being trivially
  satisfied by $B_0 = \{\varepsilon\}$.
  Let us build $B_{n+1}$ under the assumption that $B_n$ exists.
  
  Consider $\core \subseteq \iter{\clauses_{P,S}}{n}$ the $n$-th UNSAT core.
  Suppose by contradiction that for any pair $(p, a) \in B_n \times \Sigma$
  such that $\cclos{p, a} \in \core$,
  there exists $q \in B_n$ such that $pa \equiv_T q$.
  For any $(p, a) \in B_n \times \Sigma$, we define $q_{pa}$ as follows:
  \begin{itemize}
    \item if $pa \in B_n$, $q_{pa} = pa$;
    \item else if $\cclos{p, a} \in \core$ but $pa \notin B_n$, we pick some 
    $q_{pa} \in B_n$ such that $q_{pa} \equiv_T pa$;
    \item otherwise $q_{pa} = \varepsilon$—we pick an arbitrary value.
  \end{itemize}
  
  Intuitively, we use $q_{pa}$ to find representatives in $B_n$ of the pairs
  $(p, a)$ such that $\cclos{p, a} \in \core$.
  Let us now define a model $m$ induced by $B_n$ and $T$ s.t.:
  \begin{itemize}
    \item For all $ p \in \iter{P}{n}$, $m[b_p] = 1$ if and only if
    $p \in B_n$.
    \item For all
    $(p, a, p') \in \iter{P}{n} \times \Sigma \times\iter{P}{n}$,
    $m[e_{p, a, p'}] = 1$ if and only if $p \in B_n$ and $p' = q_{pa}$.
    \item For all $w \in \Sigma^*$ such that $x_w$ appears in some clause in 
    $\core$, $m[x_w] = 1$ if and only if $w \in T$.
  \end{itemize}
  
  Let us prove that $m$ satisfies every formula that can appear in $\core$:
  \begin{description}
    \item[$\cbasis{\varepsilon}$.] Since $B_n$ is non-empty and prefix-closed,
    it contains $\varepsilon$, hence $m[b_{\varepsilon}] = 1$ and
    $m \vDash \cbasis{\varepsilon}$.
    \item[$\cbasis{pa}$.] If $m[b_{pa}] = 1$, then $pa \in B_n$ by definition
    of $m$, so $p \in B_n$, $m[b_p] = 1$, and $m \vDash \cbasis{pa}$.
    \item[$\creach{pa}$.] If $m[b_{pa}] = 1$, then $pa \in B_n$ and
    $q_{pa} = pa$, thus $m[e_{p, a, pa}] = 1$ and $m \vDash \creach{pa}$.
    \item[$\ccong{p,a,p',s}$.] If $m \vDash e_{p, a, p'}$, then $p' = q_{pa}$
    and $p'\equiv_{T} pa$, hence
    $m \vDash \nu(x_{pas}) \leftrightarrow \nu(x_{p's})$
    and $m \vDash \ccong{p,a,p',s}$.
    \item[$\cdet{p, a, p_1, p_2}$] If $m[e_{p, a, p_1}] = 1$, then
    $p_1 = q_{pa}$, thus $p_2 \neq q_{pa}$, hence $m[e_{p, a, p_2}] = 0$ and 
    $m \vDash \cdet{p, a, p_1, p_2}$.
    \item[$\cclos{p, a}$.] If $m[b_p] = 1$, then $p \in B_n$ and
    $m[e_{p, a, q_{pa}}] = 1$, thus $m \vDash \cclos{p, a}$.
    \item[$\csucc{p, a, p'}$.] If $m[e_{p, a, p'}] = 1$ then $p \in B_n$, thus
    $m[b_p] = 1$ and $m \vDash \csucc{p, a, p'}$.
    \item[$\csharp{p, a, p'}$.] If $m[b_{pa}] = 1$ and $m[b_{p'}] = 1$,
    by definition of $m$, $m[e_{p, a, p'}] = 1$ if and only if $p \in B_n$ and 
    $p' = q_{pa}$.
    But if $pa \in B_n$, $q_{pa} = pa$, and $p' \neq q_{pa}$, thus
    $m[e_{p, a, p'}] = 0$ and $m \vDash \csharp{p, a, p'}$.
    \item[$\cind{w_1, w_2}$.] Since $T \in \tgt$,
    $w_1 \in T \rightarrow w_2 \in\ T$ and we also have $m \vDash \cind{w_1, w_2}$.
  \end{description}
  
  Thus, $m \vDash \core$, but $\core$ is meant to be UNSAT;
  there is a contradiction.
  
  As a consequence, there exists $p \in B_n$ and $a \in \Sigma$ such that
  $\cclos{p, a} \in \core$  and for all $q \in B_n$, $pa \not\equiv_T q$.
  In particular, $pa$ is a freshly added prefix to $\iter{P}{n+1}$, distinct
  from all elements in $B_n$.
  We therefore define $B_{n+1} = B_n \uplus \{pa\}$,
  finally proving the inductive case.\hfill\qed
\end{proof}

\subsection{Proof of Theorem~\ref{thm:termination}}
\label{prf:termination}

\begin{proof}
  Assume the algorithm did not terminate after $n-1$ refinements.
  After $n$ UNSAT core refinements, by Lemma~\ref{lem:progress}, there
  exists a prefix-closed $B_{n} \subseteq \iter{P}{n}$ such that
  $|B_{n}|= n+1$ and $\forall p, p' \in B_n$, if $ \neq p'$ then
  $p \not\equiv_T p'$, that is,
  $T$ has at least $n+1$ equivalence classes, hence,
  there is a contradiction.\hfill\qed
\end{proof}

\fi

\end{document}